\documentclass[sigconf]{acmart}

\usepackage{booktabs} % For formal tables

\usepackage{amsmath}
\usepackage{amssymb}
\usepackage{amsthm}
\usepackage{xfrac}
\usepackage{color}
\usepackage{graphicx}
\usepackage{stfloats}
\usepackage{algorithm}
\usepackage{algorithmic}
\usepackage{flushend}
\usepackage{caption}
\usepackage{enumitem}

\newtheorem{theorem}{Theorem}
\newtheorem{remark}{Remark}
\newtheorem{lemma}{Lemma}

\newcommand{\mf}{\mathbf}
\newcommand{\mc}{\mathcal}
\newcommand{\mb}{\mathbb}
\newcommand{\mr}{\mathrm}

\newlist{enumeratebyten}{enumerate}{1}
\setlist[enumeratebyten]{label={\textbf{Q\arabic*:}},ref={Q\arabic*:}}

\graphicspath{{figure/}{figures/}}

\setcopyright{acmcopyright}
\usepackage{colortbl}% http://ctan.org/pkg/xcolor
\acmDOI{10.475/123_4}

\acmISBN{123-4567-24-567/08/06}

\acmConference[MobiHoc'19]{ACM International Symposium on Mobile Ad Hoc Networking and Computing}{July 2019}{Catania, Italy} 
\acmYear{2019}
\copyrightyear{2019}

\acmPrice{15.00}

\begin{document}

\title[\textit{Jam Sessions}: Advancd Jamming Attacks in MIMO Networks]{\textit{Jam Sessions:} Analysis and Experimental Evaluation of Advanced Jamming Attacks in MIMO Networks}
% \titlenote{Produces the permission block, and copyright information}
% \subtitle{Extended Abstract}

\author{Liyang Zhang}%$^{+}$, Francesco Restuccia$^{+}$, Scott M. Pudlewski*, and Tommaso Melodia$^{+}$}

% \authornote{Note}
\affiliation{%
  \institution{Northeastern University}
%   \streetaddress{1600 Amphitheatre Parkway}
%   \city{Mountain View} 
  \city{Boston} 
  \state{MA}
    \country{USA}
  \postcode{02115}
}
\email{liyangzh@ece.neu.edu}

\author{Francesco Restuccia}
% \authornote{Note}
\affiliation{%
  \institution{Northeastern University}
%   \streetaddress{360 Huntington Ave}
  \city{Boston} 
  \state{MA}
    \country{USA}
  \postcode{02115}
}
\email{frestuc@northeastern.edu}

\author{Tommaso Melodia}
% \authornote{Note}
\affiliation{%
  \institution{Northeastern University}
%   \streetaddress{360 Huntington Ave}
  \city{Boston} 
  \state{MA} 
  \country{USA}
  \postcode{02115}
}
\email{melodia@northeastern.edu}

\author{Scott M. Pudlewski}
% \authornote{Note}
\affiliation{%
  \institution{Air Force Research Laboratory}
%   \streetaddress{360 Huntington Ave}
  \city{Rome} 
  \state{NY} 
  \country{USA}
  \postcode{02115}
}
\email{scott.pudlewski.1@us.af.mil}

\renewcommand{\shortauthors}{L.~Zhang et al.}

%!TEX root = main.tex

\begin{abstract}
Recent research advances in wireless security have shown that advanced jamming can significantly decrease the performance of wireless communications. In advanced jamming, the adversary intentionally concentrates the available energy budget on specific critical components (\textit{e.g.}, pilot symbols, acknowledgement packets, etc.) to (i) increase the jamming effectiveness, as more targets can be jammed with the same energy budget; and (ii) decrease the likelihood of being detected, as the channel is jammed for a shorter period of time. These key aspects make advanced jamming very stealthy yet exceptionally effective in practical scenarios. 

One of the fundamental challenges in designing defense mechanisms against an advanced jammer is understanding which jamming strategies yields the lowest throughput, for a given channel condition and a given amount of energy. To the best of our knowledge, this problem still remains unsolved, as an analytic model to quantitatively compare advanced jamming schemes is still missing in existing literature. To fill this gap, in this paper we conduct a comparative analysis of several most viable advanced jamming schemes in the widely-used MIMO networks. We first mathematically model a number of advanced jamming schemes at the signal processing level, so that a quantitative relationship between the jamming energy and the jamming effect is established. Based on the model, theorems are derived on the optimal advanced jamming scheme for an arbitrary channel condition. The theoretical findings are validated through extensive simulations and experiments on a 5-radio 2x2 MIMO testbed. Our results show that the theorems are able to predict jamming efficiency with high accuracy. Moreover, to further demonstrate that the theoretical findings are applicable to address crucial real-world jamming problems, we show that the theorems can be incorporated to state-of-art reinforcement-learning based jamming algorithms and boost the action exploration phase so that a faster convergence is achieved.

\end{abstract} % need to summarize experimental and theoretical results!!!

%Please copy and paste the code instead of the example below. 



% We no longer use \terms command


\maketitle

\section{Introduction} \label{sec:intro}

%Wireless transmissions have now become ubiquitous \cite{EricssonMobility2018}. For this reason, security against malicious activities will become an issue of critical importance. To this end, peculiar efforts

Wireless jamming is widely recognized as one of the most crucial topics in wireless security \cite{zou2016survey}. To understand how harmful a jammer could be, researchers have studied the worst-case jamming attack with fine-tuned temporal pattern (\textit{e.g.}, pulse jamming \cite{Wood2007}), frequency-pattern (\textit{e.g.}, frequency-hopping jamming \cite{Wu2012Cognitive}), and so on. During the last few years, a new family of \textit{advanced jamming} has gained momentum \cite{debruhl2013jam,DeBruhl-wisec2014,xiao2018reinforcement,anwar2017dynamic,pirzadeh2016subverting,yan2016jamming}, where the \emph{target component} of the jammed transmission is the main objective to optimize.

The key intuition behind advanced jamming is that, although some components of the wireless transmission do not carry payload information, they nevertheless constitute the ``Achilles' heel'' of the entire communication process. For example, at the physical (PHY) layer, \textit{pilot jamming} has been proposed to disrupt orthogonal frequency-division multiplexing (OFDM) and multiple-input and multiple-output (MIMO), since they heavily rely on accurate channel estimation through pilot symbols \cite{ClancyICC11, LaPanMILCOM12,pirzadeh2016subverting,wang2018multiple}. Furthermore, at the link (MAC) layer, acknowledgement (ACK) jamming has been proposed to disrupt medium access control operations \cite{BayrakINFOCOM08,ThuenteMILCOM06,WilhelmWiSec11}. Since advanced jamming activities are restricted to a specific period of time, an advanced jammer can degrade the network throughput with a comparatively lower energy budget yet with lower probability of being detected \cite{debruhl2013jam,DeBruhl-wisec2014}. 

The key limitation of existing work (discussed in details in Section \ref{sec:rel}) is that \textit{it does not provide the analytic tools to thoroughly investigate the quantitative relationship between the jamming energy and the jamming outcome}. As a consequence, literature still lacks a mathematical model to \emph{compare advanced jamming schemes, each with a different target component, to understand which one yields the highest jamming efficiency.} We point out that a mathematical model of advanced jamming attacks is of fundamental importance not only from a theoretical perspective, but for a number of practical reasons as well. First, it is straightforward to notice that jammers usually do not want their attacks to be discovered. Therefore, to increase stealthiness, advanced jammers need to keep the jamming signal energy below a certain threshold and corrupt the channel for as little time as possible. Furthermore, it is well known that jammed nodes usually react with strategies such as rate adaptation and rerouting \cite{Pelechrinis2009, Zhang15CNS,zhang2018taming,zhang2017learning}. Therefore, real-world advanced jammers necessarily face dynamic, time-varying scenarios, where adaptive jamming strategies are almost mandatory. In these circumstances, analytic tools on most effcient advanced jamming schemes are extremely valuable for an advanced jammer. 

As a first attempt to address the existing research gap, this article focuses on widely-used MIMO networks, and makes the following contributions:

\textbf{(1)} We select three most viable schemes in MIMO networks, each with a different target component, and rigorously model them at the signal-processing level, so that a mathematical relationship between the \emph{amount of jamming energy} used and the \emph{throughput degradation} is established for an arbitrary channel condition. Based on the model, we compare the jamming schemes and derive theorems on optimal jamming strategies in various scenarios. Our main theoretical results conclude that (i) for a given data packet, the relative efficiency for barrage jamming and pilot jamming is decided by the pilot sequence length and the number of transmitting antennae (Theorem \ref{thm1}); (ii) under some conditions, pilot jamming will lead to a lower average signal-to-inteference-and-noise ratio (SINR) of the ACK packet than directly jamming the ACK packet itself (Theorem \ref{thm2}); and (iii) ACK jamming can be compared with pilot and barrage jamming, on the basis of packet error rate (PER) lower bound, and the result is decided by a number of factors including pilot length, modulation and coding scheme (MCS) used for both data and ACK packets, and the distances from the jammer to the transmitter and receiver (Theorem \ref{thm3}). Our theoretical foundations are validated through extensive simulations and experiments on MIMO testbed made up by 5 USRP software-defined radios. Results indicate that our model is able to capture the behavior of advanced jamming strategies in complex MIMO scenarios accurately.

\textbf{(2)} To demonstrate that applicability of our theoretical results in real-world scenarios, where the theorems' assumptions cannot be validated either due to missing information or dynamic environment, we propose a way to incorporate the theorems to state-of-art reinforcement-learning \cite{sutton1998reinforcement,Restuccia-IoT2018,jagannath2019machine} based jamming algorithms, such as~\cite{Amuru2016}. Specifically, \textit{we show that the theorems on jamming efficiency can be used by the jammer to improve the efficiency of the action space exploration up to a significant extent.} Indeed, extensive simulations show the effectiveness of our approach and show significant improvement in both convergence speed and the total reward.\vspace{-0.3cm}

\subsection*{Scope and Limitations} 

We point out that the objective of our study is not to investigate every possible component in a wireless network that an advanced jammer can target. Instead, we focus on the modeling and analysis on a typical wireless scenario (MIMO networks) and several typical advanced jamming schemes, and show that (i) theoretical analysis based on rigorous model can predict the worst-case jamming results for certain scenarios; (ii) the theorems can be incorporated with practical, reinforcement-learning based jamming algorithms, thus extending their applicability to scenarios with limited information or dynamic environment. 

Moreover, we do not claim that our results provide jamming strategies that are optimal on \emph{every} possible aspect. As we have pointed out earlier, there are also other aspects of jamming strategy that can be optimized, such as temporal pattern and frequency pattern of the jamming signal. We have focused on the relatively under-explored aspect of jamming target, and provided insights on how to optimally choose the component of the wireless transmission to jam. \vspace{-0.2cm}%Furthermore, it has to be pointed out that, after the optimal jamming target is chosen, it is still possible to optimize the jamming signal following the approaches established in existing literature. We have omitted the discussion on these aspects, however, in this paper, because they are relatively independent of the optimal jamming target and has been extensively studied.

\section{Related Work} \label{sec:rel}

% Although wireless jamming has been thoroughly investigated \cite{mukherjee2014principles,zou2016survey}, only very recently has advanced jamming gained attention \cite{debruhl2013jam,DeBruhl-wisec2014,xiao2018reinforcement,anwar2017dynamic,pirzadeh2016subverting,yan2016jamming}. 

As far as the physical layer is concerned, Clancy \cite{ClancyICC11} proposes to disrupt OFDM links using pilot jamming and pilot nulling attacks. La Pan \emph{et al.} \cite{LaPanMILCOM12} consider false preamble timing and preamble nulling attacks in OFDM. Moreover, Rahbari \textit{et al.} \cite{Rahbari2016} study the impact of jamming on OFDM frequency offset (FO).
%The authors aim at creating destructive errors in the FO estimation by jamming the preamble. 
Specific to MIMO systems, Sodagari \emph{et al.} \cite{Sodagari2013} propose an attack where the jamming signal transforms the estimated MIMO channel matrix into a singular matrix. Pirzadeh \emph{et al.} show in \cite{pirzadeh2016subverting} that the MIMO spectral efficiency can be significantly degraded by jamming both the training phase and the data transmission phase. Wang \emph{et al.} \cite{wang2018multiple}  propose a random channel training (RCT)-based secure transmission framework to address MIMO jamming.

Regarding the link layer, an advanced jammer can exploit the temporal pattern between transmissions as well as critical control packets. Bayrak \emph{et al.} \cite{BayrakINFOCOM08} consider a scenario where the jammer exploits the exponential backoff mechanism of IEEE 802.11, and show that the jammer can achieve higher efficiency if aware of the current backoff states of the users. Jamming against IEEE 802.11b is discussed by Thuente \emph{et al.} in \cite{ThuenteMILCOM06}, based on launching CTS jamming, ACK jamming, or DIFS wait jamming attacks. ACK jamming in IEEE 802.15.4 is also considered by Wilhelm \emph{et al.} in \cite{WilhelmWiSec11}. %In \cite{NegiICC05}, a strategy based on combining RTS jamming and RTS faking is discussed.

The core limitation of the above mentioned pioneering works is the approach of investigating an individual jamming scheme for a fixed environment. While the potential threat of various advanced jamming schemes have been revealed, the important question of optimal advanced jamming strategy, when multiple jamming schemes are available and the environment is time-varying, is still open.

% The core limitation of the above mentioned is that optimal jamming signals can be designed provided (i) perfect knowledge on the channel between transmitter and receiver; and (ii) unlimited jamming energy budget. However, we contend that these assumptions are hardly valid in real-world scenarios where channel information is significantly time-varying and jammers aim at hiding their activities. Thus, the investigation of advanced jamming strategies requires the assumption of limited energy budget and imperfect channel knowledge. 

The closest work to ours is \cite{debruhl2013jam,DeBruhl-wisec2014}, where the authors investigate energy-optimal jamming strategies to achieve stealthiness and increase effectiveness. Specifically, DeBruhl and Tague \cite{debruhl2013jam} show that energy-efficient short-form periodic jamming can seriously degrade communication capabilities without compromising the jammer's activity. DeBruhl \emph{et al.} \cite{DeBruhl-wisec2014} investigate finite-energy jamming games, where jammers choose among different actions (i.e., sleep, power and channel). The authors compute the games' Nash equilibria, and test the performance of the optimal strategies against random and adaptive strategies. These works, however, have not focused on jamming schemes directed at critical components of a transmission, and therefore differ significantly from ours.

\vspace{-2mm}
\section{System Model} \label{sec:mod}

In this section, we illustrate a typical model for MIMO wireless communication, upon which three jamming schemes are modeled. Some commonly used notations: we use $\{*\}^T$, $\{*\}^H$ and $\{*\}^{\dagger}$ to denote the transpose, Hermitian, and pseudo inverse of a matrix $\{*\}$; $\hat{\{*\}}$ and $\tilde{\{*\}}$ represent the estimate and estimation error of $\{*\}$; $I_{\{*\}}$ is used to denote a $\{*\} \times \{*\}$ unity matrix; and $\mb{E}\{*\}$ denotes the expectation of $\{*\}$.

% \subsection{System Model}\label{sec:vulner}

{\bf Physical Layer.} We consider a MIMO link with forward (\textit{i.e.}, data) and backward (\textit{i.e.}, ACK) transmissions as illustrated in Fig. \ref{fig:mimo}.

\begin{figure}[!h]
	\centering
	\begin{tabular}{cc}
	    \includegraphics[width=0.225\textwidth]{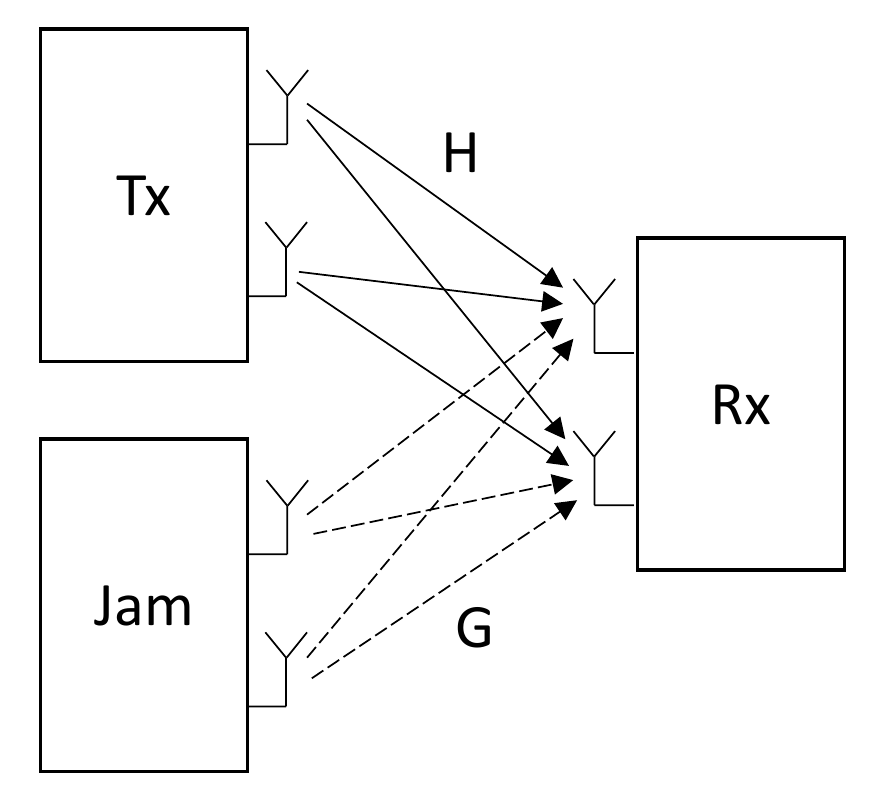} &
	    \includegraphics[width=0.225\textwidth]{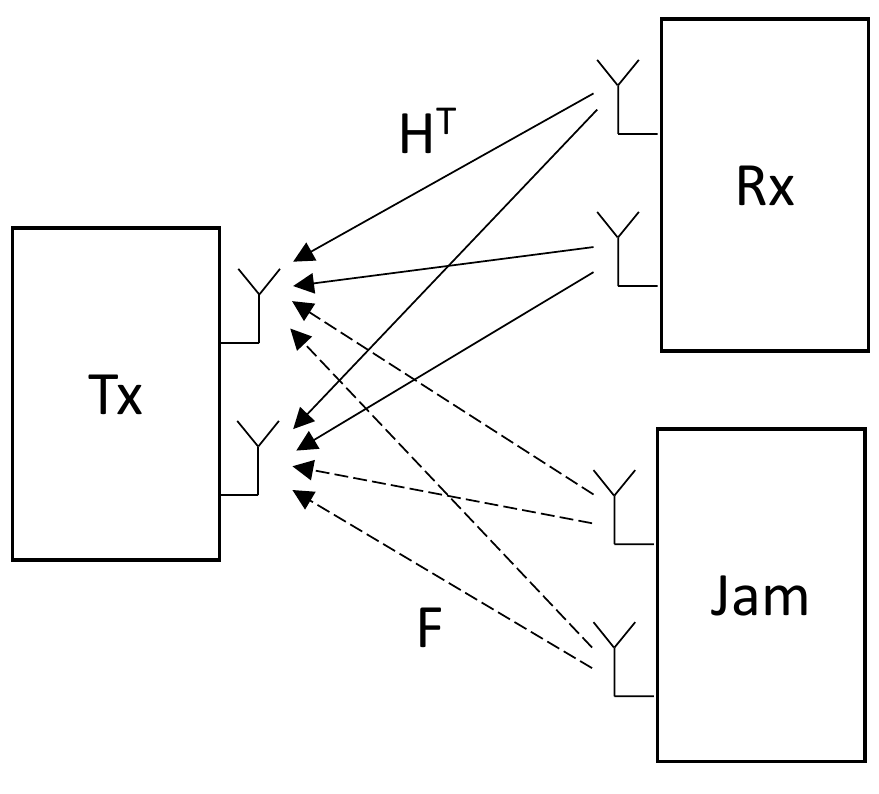} \\
	    \small (a) & \small (b)
	\end{tabular}
	\vspace{-2mm}
	\caption{\small MIMO link for (a) forward and (b) backward transmissions under jamming. Tx: transmitter; Rx: receiver; Jam: jammer.}
	\vspace{-2mm}
	\label{fig:mimo}
\end{figure}

Without loss of generality, we assume the transmitter, receiver, and the jammer are equipped with $M$, $N$, and $L$ antennae, respectively.\footnote{For simplicity, we will refer to the transmitter of the forward link as \emph{transmitter}, even when it is the receiver of the backward link. Similarly, the term \emph{receiver} will be used solely to address the receiver of the forward link.}  Therefore, the channel from the transmitter to the receiver, from the jammer to the receiver, and from the jammer to the transmitter can be denoted with matrices $H = \{h_{nm}\}_{1 \leq n \leq N, 1 \leq m \leq M}$, $G = \{g_{nl}\}_{1 \leq n \leq N, 1 \leq l \leq L}$, and $F = \{f_{ml}\}_{1 \leq m \leq M, 1 \leq l \leq L}$. Vectors $\mf{x} = \{x_1, \ldots, x_M\}^T$, $\mf{y} = \{y_1, \ldots, y_N\}^T$ and $\mf{z} = \{z_1, \ldots, z_L\}^T$ are used to represent the transmitted, received and jamming signals, respectively. while vector $\mf{w} = \{w_1, \ldots, w_N\}^T$ represents Gaussian noise at the receiver, with spectral density $N_0$.

We assume that the channels are subject to path loss and Rayleigh fading. The path loss is determined by the physical distance of the communicating parties, and therefore remains the same for every entry in the same channel matrix. We will use $\theta_{\{*\}}$ to denote the path loss of the channel $\{*\}$. Combined with the Rayleigh fading component, it follows that
\begin{eqnarray}
	h_{nm} \sim \mc{CN}(0, \theta_H), & \ \forall 1 \leq n \leq N, \forall 1 \leq m \leq M, \label{eq:h_dist}\\
	g_{nl} \sim \mc{CN}(0, \theta_G), & \ \forall 1 \leq n \leq N, \forall 1 \leq l \leq L, \label{eq:g_dist}\\
	f_{ml} \sim \mc{CN}(0, \theta_F), & \ \forall 1 \leq m \leq M, \forall 1 \leq l \leq L, \label{eq:f_dist}
\end{eqnarray}
i.e., the entries of the channel matrix are independent and identically distributed (i.i.d.) complex Gaussian random variables, with zero mean and variance equal to the path loss.

MIMO can be used to achieve multiplexing or diversity gain \cite{Zheng03}. For the data transmission on the forward link, since high throughput is usually required, we focus on spatial multiplexing, where different bits are transmitted on the $M$ antennae simultaneously (we will refer to them as $M$ spatial ``channels''). Following \cite{Proakis2007}, the baseband model for this scheme is
% For the forward transmission, we assume that spatial multiplexing is used so that the transmitter simultaneously transmits different bits on different antennae. To simplify derivation of the SINR, we also assume that channel state information (CSI) is only available at the receiver side. It follows that the baseband model is

\begin{equation}\label{eq:mimo}
	\mf{y} = \left\{
		\begin{array}{ll}
			H \mf{x} + \mf{w}, & \mr{not\ jammed,} \\
			H \mf{x} + G \mf{z} + \mf{w}, & \mr{jammed.}
		\end{array}
	\right.
\end{equation}

Let $E_s$ and $E_j$ denote the average symbol energy for the transmitted and jamming signals. Then, it follows that
\begin{align}
	\mb{E}\{\mf{x}\mf{x}^H\} & = E_s I_M, \\  
	\mb{E}\{\mf{z}\mf{z}^H\} & = E_j I_L.
\end{align}

For the ACK transmission on the backward link, since high reliability is often required, we focus on MIMO schemes that achieve diversity instead of spatial multiplexing. To this end, we consider a scheme where diversity is achieved via beamforming and Maximal Ratio Combining (also called as MIMO-MRC). Specifically, the receiver leverages the channel information acquired in the previous forward transmission, and beamforms using a vector $\mf{u}$ satisfying
\begin{equation}
	(\hat{H}^T)^H \hat{H}^T \mf{u} = \lambda_{\max} \mf{u}, %\theta_H \mf{u},
\end{equation}
\textit{i.e.}, the eigenvector corresponding to the maximal eigenvalue ($\lambda_{\max}$) of the matrix $(\hat{H}^T)^H \hat{H}^T$. A symbol $x$ is precoded by $\mf{u}$, and the resultant $\mf{u} x$ is transmitted over the antennae. Therefore, we have the baseband model
\begin{equation} \label{eq:mimo_ack}
	\mf{y} = H^T \mf{u} x + F \mf{z} + \mf{w}.
\end{equation}
Without loss of generality, we assume that $\| \mf{u} \|^2 = N$, so that the total transmitting energy across the $N$ antennae is $N E_s$, for $\mb{E} \{x^2\} = E_s$.

{\bf Link Layer.} Received data packets are acknowledged by the receiver on link layer. To establish a relationship between physical layer metrics and the PER, we borrow the model in \cite{Liu2004TWC}. Specifically, suppose a MCS $z$ is chosen from a set $\mc{Z}$, the PER $e$, as a function of SINR $\gamma$, can be approximated as
% To independently analyze packet errors on the different channels in the MIMO spatial multiplexing, we consider frame aggregation similar to \cite{Liu2004TWC}. Specifically, we assume multiple packets are aggregated into one frame and transmitted in the same transmission slot. We also assume selective automatic repeat request (ARQ), with the receiver only acknowledging the received packets. Unacknowledged packets are thus aggregated into the next frame and retransmitted.

% The transmitter may choose among a set $\mc{Z}$ of modulation and coding schemes (MCS). According to \cite{Liu2004TWC}, for a given MCS $z \in \mc{Z}$, the packet error rate (PER) $e$, as a function of SINR $\gamma$, can be approximated as
\begin{equation}\label{eq:per}
	e(\gamma | z) = \left\{
		\begin{array}{ll}
			1, & \gamma \leq \gamma_{\mr{th}}, \\
			a_z e ^{-b_z \gamma}, & \gamma > \gamma_{\mr{th}},
		\end{array}
	\right.
% \vspace{-2mm}
\end{equation}
with MCS-dependent parameters $a_z$ and $b_z$. The threshold $\gamma_{\mr{th}}$ also varies with different MCSs.% For reliability, we assume that the receiver uses low MCS to transmit ACKs.

\vspace{-2mm}
\subsection{Advanced Jamming Schemes}\label{sec:schemes}

The objective of the jammer is to efficiently degrade the throughput of the jammed link. To this end, the jammer may (i) directly inject interference to the entire data packet and lower the achievable SINR (we will refer to this scheme as \emph{barrage jamming}, following \cite{ClancyICC11}); (ii) jam the pilot symbols to invalidate the estimated channel, also called \emph{pilot jamming}  \cite{zhang2018pilot}; or (iii) prevent ACKs from being delivered, also called \emph{ACK jamming} \cite{lichtman2016communications}. We will formally model each of these cases, and analyze the resulting effects.

{\bf Barrage Jamming.} In this case, the jammer emits Gaussian noise uniformly on the entire data packet, lowering the resulting receiver SINR and consequently degrading the throughput of the jammed link. The baseband model in presence of jamming is shown in (\ref{eq:mimo}). We assume that when the jamming energy is uniformly allocated upon the entire data packet, the channel estimation error is negligible, \textit{i.e.}, $\hat{H} \approx H$. Therefore, the transmitted signal is recovered as follows:
\begin{equation}
%	\begin{array}{ll}
		\hat{\mf{x}} = \hat{H}^{\dagger} \mf{y} \approx \mf{x} + H^{\dagger} G \mf{z} + H^{\dagger} \mf{w},
%	\end{array}
\end{equation}
where $H^{\dagger} = (H^H H)^{-1} H^H$ is the left pseudo inverse of $H$.

The SINR per symbol on the $m$-th spatial channel is then
\begin{equation} \label{eq:sinr_barrage}
\small
	\gamma_m = \frac{E_s}{E_j [H^{\dagger} G G^H (H^{\dagger})^H]_{mm} + N_0 [(H^H H)^{-1}]_{mm}}.
\end{equation}

{\bf Pilot Jamming.} The advanced jammer may also aim at jamming the pilot symbols, resulting in channel estimation errors and further impairing data decoding. The pilot signal comprises of a sequence of symbols agreed on by the transmitter and receiver, denoted as $X =\{\mf{x}_1, \ldots, \mf{x}_K\}$, assuming a sequence of length $K$ is used. Similarly, the received sequence, jamming sequence, and the noise sequence are $Y = \{\mf{y}_1, \ldots, \mf{y}_K\}$, $Z = \{\mf{z}_1, \ldots, \mf{z}_K\}$, and $W =\{\mf{w}_1, \ldots, \mf{w}_K\}$, respectively. Thus, for the pilot, we have
\begin{equation}
	Y = H X + G Z + W.
\end{equation}

Least square estimation gives the estimated channel matrix
\begin{equation}\label{eq:est}
	\hat{H} = H + G Z X^{\dagger} + W X^{\dagger},
\end{equation}
with $X^{\dagger} = X^H (XX^H)^{-1}$ as the right pseudo inverse of $X$. The estimation in (\ref{eq:est}) introduces an error term
\begin{equation}
	\tilde{H} = (G Z + W) X^{\dagger},
\end{equation}
which will affect the recovery of the received signal.

%% cross-column equation that will appear in next page
%\newcounter{TempEqCnt}					% create a temp counter
%\setcounter{TempEqCnt}{\value{equation}}% assign current # to temp
%\setcounter{equation}{22}				% change current # to x, which equals to the should-be # of the equation - 1
%\begin{figure*}[hb]
%	\centering
%	\begin{equation} \label{eq:sinr_ack}
%		\gamma_{\mr{ACK}} = \frac{\| \mf{u}^H (\hat{H}^T)^H \hat{H}^T \mf{u} \|^2 E_s}{\frac{N}{K} \left[ \frac{E_{j,p}}{L} \mr{tr}(G^H G) + N N_0 \right] \mf{u}^H (\hat{H}^T)^H \hat{H}^T \mf{u} + \mf{u}^H (\hat{H}^T)^H (E_{j,a} F F^H + N_0 I_M) \hat{H}^T \mf{u}}
%	\end{equation}
%\end{figure*}
%\setcounter{equation}{\value{TempEqCnt}}% assign the temp value back to current equation #

According to \cite{Biguesh2006} and \cite{Marzetta1999}, to achieve optimal estimation, the pilot must satisfy
\begin{equation}
	XX^H = K E_s I_M.
\end{equation}
With this condition, the statistical characteristic of the channel estimation error is given by Lemma \ref{lem1}.
\begin{lemma}\label{lem1}
	The channel estimation error $\tilde{H}$ satisfies $
		\mb{E} \{\tilde{H}\} = \mf{0}_{N \times M}$
	and
	\begin{equation} \label{eq:ch_err}
		\mb{E} \{ \tilde{H} A \tilde{H}^H \} = \frac{1}{K E_s} \mr{tr}(A) (E_j GG^H + N_0 I_N),
	\end{equation}
	for an arbitrary $M \times M$ matrix $A$.
\end{lemma}

\begin{proof}
Since $\mb{E} \{\tilde{H}\} = \mf{0}_{N \times M}$ is obvious,  we will focus on $\mb{E} \{\tilde{H} A \tilde{H}\}$. With the optimal training sequence $X X^H = K E_s I_M$, we have
\begin{equation}
	\begin{split}
		\mb{E} \{\tilde{H} A \tilde{H}^H\} = & \mb{E} \{(GZ + W) X^{\dagger} A 
(X^{\dagger})^H (GZ + W)^H \} \\
		= & \frac{1}{K^2 E_s^2} \mb{E}\{(GZ + W) X^H A X (GZ + W)^H\} \\
		= & \frac{1}{K^2 E_s^2} \left( \mb{E}\{ GZ X^H A X Z^HG^H \} \right. \\
		& \left. + \mb{E}\{ W X^H A X W^H \} \right) \\
		= & \frac{1}{K^2 E_s^2} \left( D B D^H + W B W^H \right),
	\end{split}
\end{equation}
where we denote $D = GZ$, and $B = X^H A X$.

For the interference $D$, we have
\begin{equation}
	\begin{split}
		D & = G Z \\
		& = \left[
			\begin{array}{ccc}
				g_{11} & \cdots & g_{1L} \\
				\vdots & \ddots & \vdots \\
				g_{N1} & \cdots & g_{NL}
			\end{array}
		\right] \cdot \left[
			\begin{array}{ccc}
				z_{11} & \ldots & z_{1K} \\
				\vdots & \ddots & \vdots \\
				z_{L1} & \ldots & z_{LK}
			\end{array}
		\right] \\
		& = \left[
			\begin{array}{ccc}
				\sum_l g_{1l} z_{l1} & \ldots & \sum_l g_{1l} z_{lK} \\
				\vdots & \ddots & \vdots \\
				\sum_l g_{Nl} z_{l1} & \ldots & \sum_l g_{Nl} z_{lK}
			\end{array}
		\right] \\
		& \triangleq \left( \mf{d}_1, \ldots, \mf{d}_K \right).
	\end{split}
	\vspace{-2mm}
\end{equation}
Column $k$ of $Z$ represents the jamming signal at time instant $k$. Therefore, they are mutually independent. Consequently, the columns of $D$ are also uncorrelated, and
\begin{equation}
	\mb{E} \{\mf{d}_k \mf{d}_{k'}^H\} = \left\{
		\begin{array}{ll}
			E_j GG^H, & k = k', \\
			\mf{0}_{N \times N}, & k \ne k'.
		\end{array}
	\right.
\end{equation}

Therefore,

\vspace{-2mm}
\begin{equation}
	\begin{split}
		\mb{E} \{ D B D^H\} 
		= & \mb{E} \left\{ (\mf{d}_1, \ldots, \mf{d}_K) \cdot \left[
			\begin{array}{ccc}
				b_{11} & \ldots & b_{1K} \\
				\vdots & \ddots & \vdots \\
				b_{K1} & \ldots & b_{KK}
			\end{array}
		\right] \right. \left. \cdot \left(
			\begin{array}{c}
				\mf{d}_1^H \\
				\vdots \\
				\mf{d}_K^H
			\end{array}
		\right) \right\} \\
		= & \mb{E} \left\{\sum_k b_{kk} \mf{d}_k \mf{d}_k^H\right\} =  E_j \mr{tr}(B) GG^H.
	\end{split}
	\vspace{-2mm}
\end{equation}

Since $B = X^H A X$, we have $\mr{tr}(B) = \mr{tr}(X^H A X) = \mr{tr}(X X^H A) = K E_s \mr{tr}(A)$
and $\mb{E} \{D B D^H\} = E_j K E_s \mr{tr}(A) GG^H$.

For the term $WBW^H$, since all the entries of $W$ are i.i.d, we have
\begin{equation}
	\mb{E} \{W B W^H\} = N_0 \mr{tr}(B) = N_0 K E_s \mr{tr}(A) I_N.
\end{equation}

Therefore,
\begin{equation}
	\mb{E} \{\tilde{H} A \tilde{H}^H\} = \frac{1}{K E_s} \mr{tr}(A) \left( E_j GG^H + N_0 I_N \right).
\end{equation}
\end{proof}

The subsequent data transmission is not affected by jamming, and complies with (\ref{eq:mimo}). Therefore, the decoded signal is 
\begin{equation}
	\hat{\mf{x}} = \hat{H}^{\dagger} \mf{y} \approx \mf{x} + H^{\dagger} \mf{w} - H^{\dagger} \tilde{H} \mf{x} -
H^{\dagger} \tilde{H} H^{\dagger} \mf{w}.
\end{equation}
%\begin{equation}
%	\begin{split}
%		\hat{\mf{x}} & = \hat{H}^{\dagger} \mf{y} \\
%		& \approx (H^{\dagger} - H^{\dagger} \tilde{H} H^{\dagger}) (H
%\mf{x} + \mf{w}) \\
%		& = \mf{x} + H^{\dagger} \mf{w} - H^{\dagger} \tilde{H} \mf{x} -
%H^{\dagger} \tilde{H} H^{\dagger} \mf{w}.
%	\end{split}
%\vspace{-2mm}
%\end{equation}
The post-processing noise is then
\begin{equation}\label{eq:post_noise}
	\hat{\mf{w}} = H^{\dagger} \mf{w} - H^{\dagger} \tilde{H} \mf{x} -
H^{\dagger} \tilde{H} H^{\dagger} \mf{w},
\end{equation}
with an autocorrelation given by Lemma \ref{lem2}.
\begin{lemma} \label{lem2}
	The Autocorrelation of the post-processing noise (\ref{eq:post_noise}) is
	\begin{equation} \label{eq:pilot_noise}
	\small
		\hspace{-1mm}\mb{E} \{\hat{\mf{w}}\hat{\mf{w}}^H\} \approx \left(1 + \frac{1}{K}\right) N_0 (H^HH)^{-1} + E_j \frac{M}{K} H^{\dagger} GG^H (H^{\dagger})^H.\hspace{-3mm}
	\end{equation}
\end{lemma}

\begin{proof}
    The autocorrelation can be derived as follows:
\begin{equation}
	\begin{split}
		\mb{E} \{\hat{\mf{w}}\hat{\mf{w}}^H\} = & \mb{E} \{(H^{\dagger} \mf{w} - H^{\dagger} \tilde{H} \mf{x} - H^{\dagger} \tilde{H} H^{\dagger} \mf{w}) \\
		& \cdot (H^{\dagger} \mf{w} - H^{\dagger} \tilde{H} \mf{x} - H^{\dagger} \tilde{H} H^{\dagger} \mf{w})^H\} \\
		=~& \mb{E} \{ H^{\dagger} \mf{w} \mf{w}^H (H^{\dagger})^H + H^{\dagger} \tilde{H} \mf{x} \mf{x}^H \tilde{H}^H (H^{\dagger})^H \\
		& + H^{\dagger} \tilde{H} H^{\dagger} \mf{w} \mf{w}^H (H^{\dagger})^H \tilde{H}^H (H^{\dagger})^H \} \\
		=~& N_0 (H^HH)^{-1} + \frac{M}{K} H^{\dagger} \left( E_j GG^H + N_0 I_N \right) (H^{\dagger})^H \\
		& + \frac{N_0}{KE_s} \mr{tr}( (H^HH)^{-1} ) H^{\dagger} ( E_j GG^H + N_0 I_N ) (H^{\dagger})^H
	\end{split}
\end{equation}

According to \cite{Wang2007TWC}, $\mr{tr}((H^HH)^{-1})$ is a small value typically no larger than $M$. Since we are focusing on high SNR (by SNR we mean signal to noise ratio, the SIR or SJR is not necessarily high) scenario, this implies that $\sfrac{N_0}{E_s} \mr{tr}((H^HH)^{-1}) << M$. In other words, the last term is much smaller than the second term and can be neglected. Therefore, we have
\begin{equation}
	\mb{E} \{\hat{\mf{w}}\hat{\mf{w}}^H\} \approx \left(1 + \frac{1}{K}\right) N_0 (H^HH)^{-1} + \frac{M}{K} E_j H^{\dagger} GG^H (H^{\dagger})^H.
\end{equation}
\end{proof}
	
Therefore, the ``effective'' SINR on the $m$-th spatial channel is
\begin{equation} \label{eq:sinr_pilot}
\small
%\begin{split}
	\gamma_m = \frac{E_s}{E_j \frac{M}{K} [H^{\dagger} GG^H (H^{\dagger})^H]_{mm} + \left(1 + \frac{1}{K}\right) N_0 [(H^HH)^{-1}]_{mm}}, 
%	\end{split}
\end{equation}

{\bf ACK Jamming.} With the baseband model for ACK transmission in (\ref{eq:mimo_ack}), the signal can be recovered by multiplying $\mf{u}^H(\hat{H}^T)^H$ with the received signal and normalizing. Therefore, at the transmitter side, we have
\begin{equation}
	\begin{split}
		\hat{x} & = \frac{\mf{u}^H (\hat{H}^T)^H (H^T \mf{u} x + F \mf{z} + \mf{w})}{\mf{u}^H (\hat{H}^T)^H \hat{H}^T \mf{u}} \\
		& = x - \frac{\mf{u} (\hat{H}^T)^H \tilde{H}^T \mf{u} x - \mf{u} (\hat{H}^T)^H (F \mf{z} + \mf{w})}{\mf{u}^H (\hat{H}^T)^H \hat{H}^T \mf{u}},
	\end{split}
\end{equation}
and the post-processing SINR is as shown in (\ref{eq:sinr_ack}), with $E_{j,p}$ and $E_{j,a}$ denoting the jamming energy per symbol in pilot jamming and ACK jamming.
\begin{equation} \label{eq:sinr_ack}
    \begin{split}
    \gamma_{\mr{ACK}} = \| \mf{u}^H (\hat{H}^T)^H \hat{H}^T \mf{u} \|^2 E_s / \{ \left[ \frac{E_{j,p}}{L} \mr{tr}(G^H G) + N N_0 \right] \\
	\cdot \mf{u}^H (\hat{H}^T)^H \hat{H}^T \mf{u} + \mf{u}^H (\hat{H}^T)^H (E_{j,a} F F^H + N_0 I_M) \hat{H}^T \mf{u} \}
    \end{split}
\end{equation}

% \begin{figure*}[b]
% \begin{equation} \label{eq:sinr_ack}
% \centering
% 	\gamma_{\mr{ACK}} = \frac{\| \mf{u}^H (\hat{H}^T)^H \hat{H}^T \mf{u} \|^2 E_s}{\left\{ \left[ \frac{E_{j,p}}{L} \mr{tr}(G^H G) + N N_0 \right] \cdot \mf{u}^H (\hat{H}^T)^H \hat{H}^T \mf{u} + \mf{u}^H (\hat{H}^T)^H (E_{j,a} F F^H + N_0 I_M) \hat{H}^T \mf{u} \right\}}
% \end{equation}
% \end{figure*}

\section{Theoretical Analysis of Optimal Jamming Target} \label{sec:ana}

In this section, we will quantitatively compare the jamming schemes modeled in Section \ref{sec:mod} for the optimal target in the sense of jamming efficiency. Since different targets (\textit{i.e.}, data packet, pilot, and ACK packet) differ in terms of their lengths, it is unfair to directly compare the jamming power. Instead, we compare the jamming effects caused by a \emph{unit} of jamming energy spent on each specific target, regardless of their lengths (\textit{e.g.}, we compare the effect of spending one unit of energy on pilot jamming vs ACK jamming).

\vspace{-2mm}
\subsection{Physical Layer Jamming}
At the physical layer, the objective of degrading throughput is equivalent to reducing the achievable SINR, for which the following theorem holds.

\begin{theorem}\label{thm1}
	For a unit of jamming energy, pilot jamming yields a lower SINR than barrage jamming if
	\begin{equation} \label{eq:cond1}
	    K < \sqrt{D\cdot M},
	\end{equation}
	where $D$ and $K$ are the lengths of the data packet and pilot in symbols, and $M$ is the number of antennae at the transmitter.
\end{theorem}

\begin{proof}
	For a unit of jamming energy, the values of jamming energy per symbol can be expressed as $E_j = 1 / (L\cdot K)$ and $E_j = 1 / (L \cdot D)$, with $L$ as the number of antennae at the jammer, for pilot jamming and barrage jamming, respectively. Plugging these into the SINR expressions in (\ref{eq:sinr_barrage}) and (\ref{eq:sinr_pilot}), the theorem follows immediately.
\end{proof}

\begin{remark}
\label{rem1}
    According to Theorem \ref{thm1}, for a given data packet and the number of transmitting antennae, the optimal choice between barrage jamming and pilot jamming is decided solely by the pilot sequence length. This matches intuition, since longer pilot sequences are more robust to jamming. As long as the pilot length $K \geq \sqrt{DM}$, there is no incentive for the jammer to launch a pilot jamming attack, and the optimal jamming strategy is reduced to barrage jamming.
    
    Note, here we are comparing on the basis that every data packet has a pilot. In reality, a pilot is supposed to cover a channel coherence period, during which there might be multiple data packets transmitted. In this case, the data packets can be treated as one single payload, in the sense of pilot jamming, since they share the same pilot.
\end{remark}

\begin{remark}
\label{rem2}
    The square-root form of data packet length $D$ in (\ref{eq:cond1}) may look anti-intuitive. A better interpretation is available if we rewrite (\ref{eq:cond1}) to the equivalent form of
    \begin{displaymath}
        \frac{M}{K} \cdot \frac{1}{K} > \frac{1}{D}.
    \end{displaymath}
    Note, with a unit jamming energy, $\frac{1}{K}$ and $\frac{1}{D}$ are the per-symbol jamming energy for pilot and barrage jamming, respectively. So Theorem \ref{thm1} essentially means that, the interference caused on the pilot signal is equivalently scaled by $\frac{M}{K}$ during the signal recovering phase for the following data packet. (For a in-depth understanding of this scale, please refer to the proofs of Lemma \ref{lem1} and \ref{lem2}.)
    
    This scaling reflects the way how pilot jamming works. The jamming signal added to the pilot introduces errors to the estimated channel matrix, which are ``transformed'' during the signal processing in signal recovering phase. As a result, for the recovered signal, the interference is equivalently scaled by $\frac{M}{K}$. Interestingly, the scaling effect is proportional to the number of transmitting antennae, and inversely proportional to the pilot length. Therefore, while increasing number of transmitting antennae increases the total throughput, it does not alleviate pilot jamming. On the contrary, it makes the situation worse.

\end{remark}

\vspace{-2mm}
\subsection{Link Layer Jamming}
At the link layer, the jammer aims at disrupting the transmission of ACK packets. Since a corrupted ACK leads to retransmission of the corresponding data packet, it is equivalent to a corrupted data packet in terms of the effective link throughput. To evaluate the jamming effect, we need to compare the PER for the data and ACK packets. However, substituting the complex forms of \emph{instantaneous} SINR of (\ref{eq:sinr_barrage}), (\ref{eq:sinr_pilot}) and (\ref{eq:sinr_ack}) into the PER-SINR function (\ref{eq:per}) produces intractable equations. To ease the analysis, we apply the \emph{expected} SINR in (\ref{eq:per}) instead, which produces a lower bound for the PER, as shown below.

%At the link layer, the advanced jammer aims at disrupting the transmission of ACK packets. This is because a corrupted ACK leads to retransmission of the corresponding data packet, effectively reducing the overall link throughput. To evaluate the jamming effect, PER is needed. However, the PER-SINR relation is generally complex for most schemes, even with adequate approximations as in (\ref{eq:per}). The issue is further exacerbated by the complex forms of the SINR as a function of jamming energy per symbol, as shown in Eq. (\ref{eq:sinr_barrage}), (\ref{eq:sinr_pilot}), and (\ref{eq:sinr_ack}). Thus, we analyze the \emph{expected} SINR for different jamming targets, which provides a lower bound for the PER.

With Rayleigh fading model, the entries of channel matrices $H$, $G$, and $F$ are i.i.d. complex Gaussian random variables, as shown in Eq. (\ref{eq:h_dist}), (\ref{eq:g_dist}), and (\ref{eq:f_dist}). The expected SINRs for the three jamming schemes in Eq. (\ref{eq:sinr_barrage}), (\ref{eq:sinr_pilot}), and (\ref{eq:sinr_ack}) are then derived as
\begin{align}
	\bar{\gamma}_{\mr{b}} & = \frac{E_s \cdot  \theta_H}{E_{j,a} \cdot \theta_G + N_0}, \label{eq:sinr_av_barrage} \\
	\bar{\gamma}_{\mr{p}} & = \frac{E_s \cdot \theta_H}{\frac{M}{K}\cdot  E_{j,p} \cdot \theta_G + (1 + \frac{1}{K}) \cdot N_0}, \label{eq:sinr_av_pilot} \\
	\bar{\gamma}_{\mr{a}} & = \frac{\mb{E}\{\lambda_{\max}\}\cdot  N \cdot  E_s \cdot \theta_H}{\frac{N^2}{K} (E_{j,p} \cdot \theta_G + N_0) + (L \cdot E_{j,a} \cdot  \theta_F + N_0)}, \label{eq:sinr_av_ack}
\end{align}
where we use subscripts $b$, $p$, and $a$ in $E_{j, \{*\}}$ to distinguish the per-symbol jamming energy for barrage jamming, pilot jamming, and ACK jamming, respectively.
%Note that the SINR for barrage jamming and pilot jamming differs for different spatial channels in Eq. (\ref{eq:sinr_barrage}) and (\ref{eq:sinr_pilot}). However,  the expected SINR is the same for all the channels, as we assume i.i.d. channels. 
With the PER model in (\ref{eq:per}), we approximate the PER as
\begin{equation} \label{eq:per_ineq}
	e \approx 1 - \mb{P}\{\gamma \geq \gamma_{\mr{th}}\} \geq 1 - \frac{\bar{\gamma}}{\gamma_{\mr{th}}},
\end{equation}
where the inequality comes from Markov's inequality. Apparently, (\ref{eq:per_ineq}) establishes a PER lower bound
\begin{equation} \label{eq:per_lower}
    e_{\mr{lwr}} = 1 - \frac{\bar{\gamma}}{\gamma_{\mr{th}}}.
\end{equation}

\begin{theorem} \label{thm2}
    When used exclusively, pilot jamming yields a higher PER lower bound for the following ACK packet than ACK jamming, with a unit jamming energy, if the following holds:
	\begin{equation} \label{eq:cond2}
		\frac{A \cdot \theta_G}{K \cdot  \theta_F} \geq \frac{K \cdot L}{N^2}
	\end{equation}
    where $A$ and $K$ are lengths of the ACK packet and the pilot in symbols; $\theta_G$ and $\theta_F$ denote the path loss from the jammer to the receiver and transmitter; $L$ and $N$ are the number of antennae at the jammer and the receiver, respectively.
\end{theorem}

\begin{proof}
	 Since we are comparing the PER lower bound of the same packet (ACK), the term $\gamma_{\mr{th}}$ in (\ref{eq:per_lower}) holds the same for both jamming schemes. Therefore, we only need to compare the expected SINR in (\ref{eq:sinr_av_ack}). 
	 
	 For a unit of jamming energy, the jamming energy per symbol is $1 / (L\cdot K)$ and $1 / (L \cdot A)$, for pilot jamming and ACK jamming, respectively. Plugging these into the first and second terms in the denominator of Eq. (\ref{eq:sinr_av_ack}) and ignore the noise, it follows that, when (\ref{eq:cond2}) holds:
	\begin{displaymath}
		\frac{N^2}{K} \cdot E_{j,p}\cdot  \theta_G > L\cdot  E_{j,a} \cdot \theta_F,
	\end{displaymath}
	\textit{i.e.}, when used exclusively, pilot jamming yields a higher denominator than ACK jamming. A lower (\ref{eq:sinr_av_ack}) follows. Hence a higher (\ref{eq:per_lower}).
\end{proof}

\begin{remark}
\label{rem2}
    Theorem \ref{thm2} states that when the condition holds, pilot jamming alone will lead to a higher PER lower bound of the ACK packet than directly jamming the ACK packet itself. This is caused by incorrect CSI being propagated directly to the ACK transmission.
    
    Note the primary target of pilot jamming is the data packet. However, as long as the conditions in Theorem \ref{thm2} holds, the optimal jamming scheme to disrupt the following ACK packet is also pilot jamming (on the prior data packet). Therefore, under this condition, ACK jamming cannot be the optimal jamming scheme, and no jamming energy should be spent on it. Partly based on this theorem, the following theorems will give the optimal jamming schemes.
    
    %An important factor affecting the relative comparison is the path loss ratio $\theta_G / \theta_H$, suggesting the jammer may need to adjust its strategy when rerouting happens.
\end{remark}

\begin{theorem} \label{thm3}
	Let us assume that noise is negligible.  With a unit of  jamming energy, ACK jamming yields a higher PER lower bound than barrage jamming if
	\begin{equation} \label{eq:cond3}
		\frac{A \cdot \theta_G}{D \cdot \theta_F} < \frac{L \cdot \gamma_{\mr{th},a}}{\mb{E}\{\lambda_{\max}\} \cdot N \cdot  \gamma_{\mr{th},d}};
	\end{equation}
	and higher lower bound than pilot jamming if
	\begin{equation} \label{eq:cond4}
		\frac{A \cdot \theta_G}{K \cdot \theta_F} < \frac{K \cdot L \cdot  \gamma_{\mr{th},a}}{\mb{E}\{\lambda_{\max}\} \cdot M \cdot  N \cdot \gamma_{\mr{th},d}},
	\end{equation}
	where $A$, $D$, and $K$ are the lengths of the ACK packet, data packet, and pilot, respectively; $M$, $N$, and $L$ are the numbers of antennae at the transmitter, receiver, and  jammer, respectively; $\theta_F$ and $\theta_G$ are the path loss from the jammer to the transmitter and receiver; $\gamma_{\mr{th},d}$ and $\gamma_{\mr{th},a}$ are the SINR thresholds for data and ACK packets. $\lambda_{\max}$ is the maximum eigenvalue of $(\hat{H}^T)^H \cdot \hat{H}^T$.
\end{theorem}

\begin{proof}
	With a unit of jamming energy, the jamming energy per symbol is $E_{j,b} = 1 / (L D)$, $E_{j,p} = 1 / (L K)$, and $E_{j, a} = 1 / (L A)$, with $L$ as the number of antennae at the jammer, for barrage jamming, pilot jamming, and ACK jamming, respectively. Omitting the noise components and plugging them in (\ref{eq:sinr_av_barrage}), (\ref{eq:sinr_av_pilot}) and (\ref{eq:sinr_av_ack}), and applying them to (\ref{eq:per_lower}), the theorem follows.
\end{proof}

\begin{remark}
	Although (\ref{eq:cond4}) is sufficient to guarantee a higher PER lower bound of the ACK packet (as a result of ACK jamming) than that of the data packet (as a result of pilot jamming), it does not guarantee that ACK jamming is the best choice. Recall Theorem \ref{thm2} and Remark \ref{rem2}, pilot jamming also affects the PER of the following ACK packets, and sometimes yields a higher PER lower bound (on the ACK packet) than ACK jamming per unit jamming energy. Therefore, only when conditions in (\ref{eq:cond2}) and (\ref{eq:cond4}) both hold, is ACK jamming a better choice than pilot jamming.
\end{remark}

\begin{remark}
	The major factors affecting (\ref{eq:cond3}) and (\ref{eq:cond4}) include the SINR thresholds $\theta_{\mr{th}}$ for the MCSs, the component lengths $K$, $D$ and $A$, and the path loss $\theta_{\{*\}}$ for $G$ and $F$. This means the optimal jamming target varies with the MCSs used by the data and ACK packets, the pilot, payload, and ACK lengths, as well as the distance from the jammer to the transmitter/receiver. In a scenario where any of these factors are dynamic (e.g., in IEEE 802.11 with rate adaptation), optimal jamming strategy needs to be dynamic, too.
\end{remark}

\vspace{-2mm}
\section{Application to Practical Scenarios} \label{sec:alg}

The theorems presented in the previous section reveals that the optimal jamming target is decided by a set of variables. Obviously, if perfect information about these variables is available, the jammer can at all times optimize its strategy to maximize the impairment generated to the legitimate communication process. However, in practical scenarios, information describing these variables (\textit{e.g.}, pilot length, path loss, etc.) may be unknown or dynamic. Therefore, it remains unclear whether the theoretical results are applicable to practical scenarios. To answer this question, we resort to machine learning and show that the theorems can be incorporated to and enhance learning-based algorithms.

In the following, we will first introduce a variant of state-of-art reinforcement-learning based jamming algorithm; then, we design a novel method to improve the algorithm by using the theorems to boost the action exploration. In this way, we provide a way to apply the theoretical results to practical scenarios with unknown and dynamic environment.

\subsection{Reinforcement-Learning Based Jamming}
Reinforcement learning \cite{sutton1998reinforcement} allows the learning agent to adapt to the optimal action to maximize the reward in a certain environment through trial-and-error. It has drawn attentions from researchers and jamming algorithms based on it have been proposed in \cite{Amuru2016,ZhuanSun2017}. We will derive a similar reinforcement learning algorithm for a dynamic and interactive scenario described below.

Without loss of generality, we consider that before each transmission, the transmitter may choose (i) a MCS; and (ii) a route. The choice of MCS is guided by a rate adaption scheme that adjusts the MCS according to the achieved PER. The transmitter may also %choose to reroute, i.e., it may 
choose among a set of possible receivers for the next hop, depending on the link quality to each of them. As a result, the jammer faces a dynamic unknown environment where information such as MCS and path loss changes interactively with its own actions.

\begin{figure}
	\centering
	\includegraphics[width=0.4\textwidth]{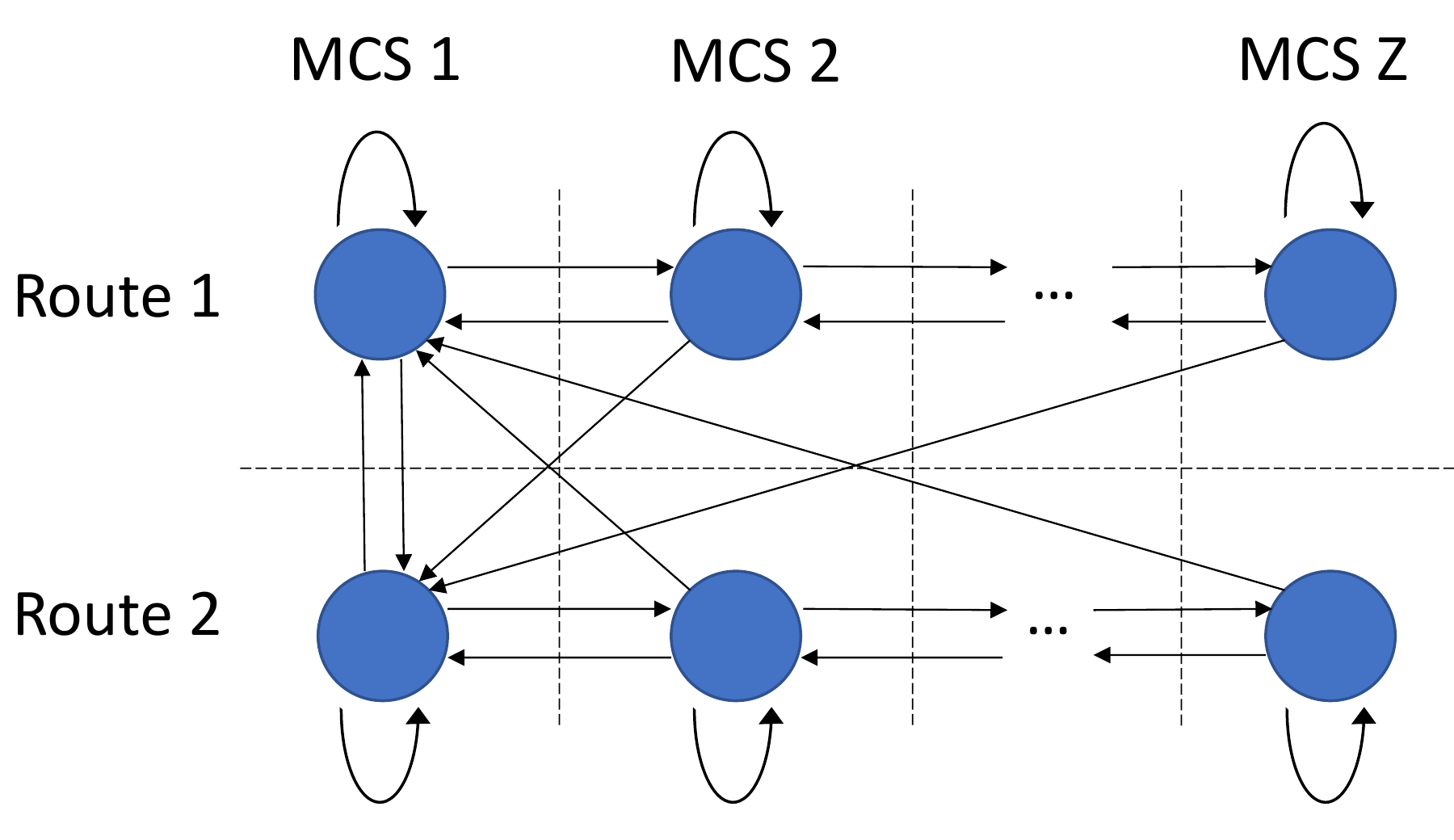}
	\caption{\small The advanced jamming problem as an MDP.}
	\vspace{-0.5cm}
	\label{fig:states}

\end{figure}

The scenario can be modeled as a Markov Decision Process (MDP), as shown in Fig. \ref{fig:states}
. Formally, it can be defined as a 5-tuple $\{\mc{S}, \mc{A}, \mb{P}(s,s'), \allowbreak R(s,a), \beta\}$, where $\mc{S}$ and $\mc{A}$ represent the \textit{state} set and \textit{action} set, $\mb{P}(s,s')$ is the transition probability from state $s$ to $s'$, $R(s,a)$ is the reward obtained by the decision agent for taking action $a$ at state $s$, and $\beta$ is a discount factor. 

In our model, the state set is the Cartesian product of the MCS set and route set. To capture the different jamming schemes discussed in Sections \ref{sec:mod}, we set the action space to a tuple $a = \{J_b, J_p, J_a, T_p\}$, where $J_b$, $J_p$, and $J_a$ denote the energy spent on barrage jamming, pilot jamming, and ACK jamming, respectively. We denote by $T_p$ the length of the pilot jamming signal, since the real length of the pilot $K$ is not directly observable.  

% However, given the lengths of the data packet and the ACK packet can be easily estimated through observation, the lengths of the barrage jamming and ACK jamming signals do not constitute a degree of freedom in the action space.

The objective of the jammer is to reduce the wireless node's throughput with maximum energy efficiency. Therefore, the reward should increase with the throughput degradation and decrease with the energy consumption. For simplicity, we use the degraded throughput subtracted by a ``price'' paid for the energy spent for the degradation as the reward. Specifically, $R(s, a) = n_e(a) - p (J_b + J_p + J_a)$, where $n_e(a)$ is the number of data packets lost as a result of the jamming action $a$ and $p$ is the price for a unit of energy. The objective of the jammer is to identify a policy $\pi: s \rightarrow a, \forall s \in \mc{S}, \forall a \in \mc{A}$, that maximizes the discounted sum of the instant and future rewards. We solve this Markov decision process (MDP) problem using the well-known Q-learning algorithm \cite{sutton1998reinforcement}, in which the optimum policy is found by iteratively updating the value function for each combination of state and action. We will omit the detailed learning algorithm here for space limit, and focus on the novel action exploration scheme.\vspace{-0.2cm}

\subsection{Theorem-Enhanced Action Exploration}
Reinforcement learning algorithms are guaranteed to converge, as long as certain conditions are met, \textbf{but the convergence speed heavily depends on the action exploration method \cite{sutton1998reinforcement}.} Indeed, at each iteration, the learning engine may choose either to (i) \textit{explore} the under-explored actions, or (ii) \textit{exploit} the actions already known to yield a good reward. Thus, an effective exploration method may improve the convergence of the learning process significantly \cite{Dearden1998}. \textbf{Since the theorems derived in Section \ref{sec:ana} establish that certain actions are more favorable than others in certain states, we use them to improve the learning exploration phase}.

According to Theorem \ref{thm1}, the jammer should limit the action space to pilot jamming with $T_p \leq \sqrt{D_{\max} \cdot M}$ and barrage jamming. Moreover, Theorem \ref{thm1} also provides a criterion to choose between pilot jamming and barrage jamming. Suppose the jammer has an estimate on both the data packet length $D$ and pilot length $K$, in the form of probability distribution $\mb{P} \{K = k\}$. Then, it can compute the probability that pilot jamming is more efficient than barrage jamming, as
\begin{equation} \label{eq:prob_b_jam}
	\mb{P}_b = \sum_{d \in \mc{D}, k \in \mc{K}} \mb{P} \{D = d\} \cdot \mb{P} \{k > \sqrt{d M}\}.
\end{equation}
Consequently, the jammer should choose barrage jamming with probability $\mb{P}_b$, and pilot jamming with $T_p = k$ with probability
\begin{equation} \label{eq:prob_p_jam}
	\mb{P}_p = (1 - \mb{P}_b) \cdot \mb{P} \{K = k\}.
\end{equation}
% The distribution $\mb{P} \{K = k\}$ can be updated in different ways. A simple one is to utilize the jamming result. If a packet is successfully jammed with $T_p = k$, add $1$ point to the value $k$. Then compute the probability using the points for different $k$.

At the link layer, Theorem \ref{thm2} and Theorem \ref{thm3} specify conditions for the jammer to use ACK jamming, and can be used to prune unfavorable pilot jamming or barrage jamming actions. To be specific, when
\begin{equation} \label{eq:cond6}
	\frac{A \cdot \theta_G}{K \cdot \theta_F} < \min \left( \frac{K \cdot L}{N^2}, \frac{K \cdot L \gamma_{\mr{th},a}}{\mb{E}\{\lambda_{\max}\} \cdot M \cdot N \cdot \gamma_{\mr{th},d}} \right)
\end{equation}
holds, ACK jamming is more favorable than pilot jamming; and when Eq. (\ref{eq:cond3}) holds, ACK jamming is more favorable than barrage jamming. Since Theorem \ref{thm2} and \ref{thm3} are valid for lower bounds on the PER,  we let the jammer explore ACK jamming when the conditions hold with probability $\epsilon$. Variables $\theta_G$, $\theta_F$ can be estimated by averaging the signal strength received from the transmitter and receiver, given the knowledge of transmitting power. $\mb{E} \{\lambda_{\max}\}$, $\gamma_{\mr{th},a}$, and $\gamma_{\mr{th},d}$ can be computed offline. To summarize, the procedure for enhanced exploration is described as a decision tree in Fig. \ref{fig:decision_tree}, and reported in detail in Algorithm \ref{alg:exp}.

\begin{figure}
    \centering
    \includegraphics[width=0.45\textwidth]{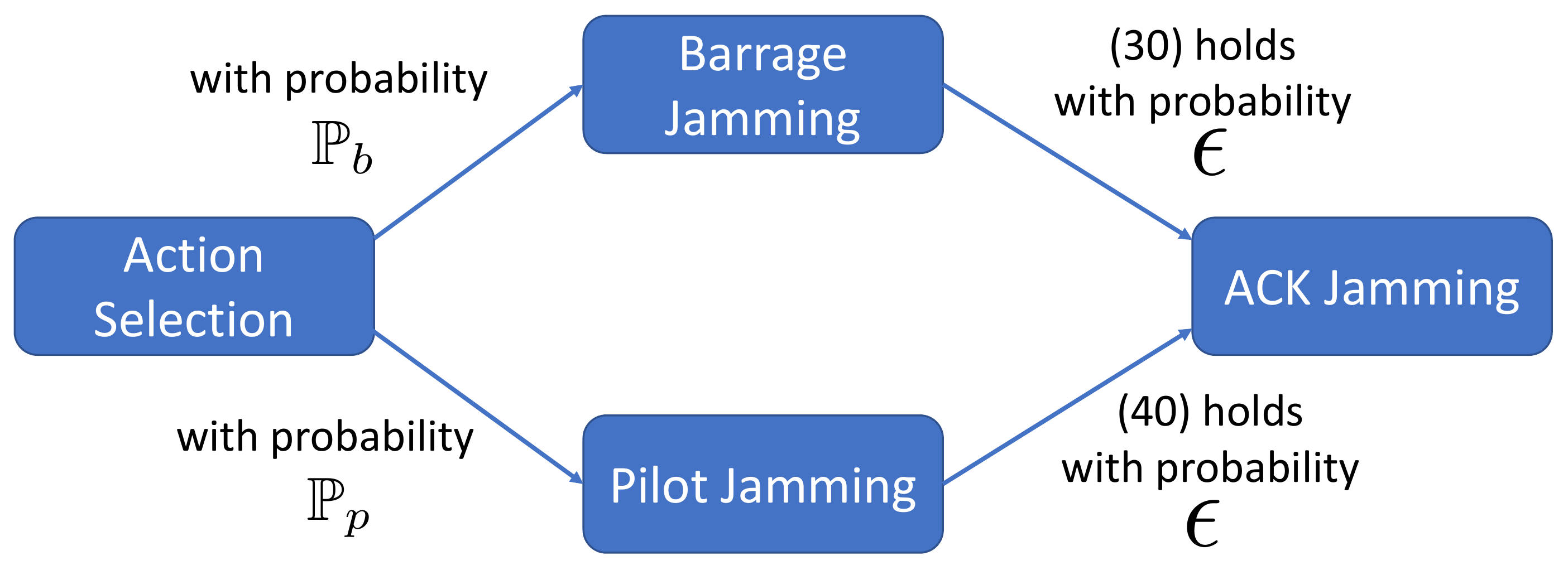}
    \caption{\small Decision tree for action selection.}
    \vspace{-2mm}
    \label{fig:decision_tree}
\end{figure}

\begin{algorithm}[h!]
\small
	\caption{Enhanced Exploration Algorithm}
	\label{alg:exp}
	\begin{algorithmic} 
		\STATE Initialize $\mb{P}\{K = k\}, k < \sqrt{D_{\max} M}$;
		\FOR {each iteration of Q-Learning Algorithm}
		\STATE Choose barrage jamming with probability (\ref{eq:prob_b_jam}), and $T_p = k$-length pilot jamming with probability (\ref{eq:prob_p_jam});
		\IF {barrage jamming is chosen and (\ref{eq:cond3}) holds}
		\STATE Choose ACK jamming with probability $\epsilon$;
		\ENDIF
		\IF {$T_p = k$-length pilot jamming is chosen and (\ref{eq:cond6}) holds}
		\STATE Choose ACK jamming with probability $\epsilon$;
		\ENDIF
		\STATE Choose jamming energy and set the action;
		\STATE Update $\mb{P}\{K = k\}$ if pilot jamming is chosen, according to the jamming result;
		\ENDFOR
	\end{algorithmic}
\end{algorithm}
\normalsize

% Note that the jammer also needs to decide the energy that it is willing to spend on each chosen target, \textit{i.e.}, $J_b$, $J_p$, and $J_a$. \textit{Therefore, the jamming energy is also a component of the action}. However, the theorems we have used for action exploration establish rules to choose the best jamming target for a given amount of jamming energy. Therefore, the exploration of jamming energy is independent from  exploration of the jamming target. State-of-the-art exploration methods, such as semi-uniform random exploration and Boltzmann exploration, can be used for energy exploration.

\section{Performance Evaluation}

We first validate our theorems in Section \ref{sec:validation}, followed by a testbed evaluation in Section \ref{sec:testbed} and by an evaluation of the learning-based jamming algorithm in Section \ref{sec:eva}. \vspace{-0.2cm}

\subsection{Theoretical Validation}\label{sec:validation}

To validate the theorems in Section \ref{sec:ana}, we simulate advanced jamming schemes on a MIMO link. We consider a scenario where the transmitter, receiver, and the jammer all equipped with $2$ antennae. The data packet size is set to $1024\:\mr{bits}$, and $240$ packets are aggregated to one frame, i.e., with one pilot. The pilot length may be $4$, $16$, $128$, or $512$ symbols long. ACK packets are assumed to be $512\:\mr{bits}$ long. The data packets can be modulated with BPSK, QPSK, 16QAM, or 64QAM, while ACK packets are modulated with BPSK.
%, and $240$ packets are aggregated into one frame. Therefore, in one transmission slot, $120$ packets are transmitted over each of the transmitting antennae. 

% The frame size in symbols for each of the modulations is shown in Table. \ref{tbl:payload}. 

We evaluate barrage jamming, pilot jamming, and ACK jamming with normalized jamming energy between $0$ and $20$, where the energy is normalized by the \emph{energy per transmitted symbol}. Note, for barrage jamming on BPSK-modulated signals, a jamming energy of $20$ is equivalent to a signal-to-jamming ratio (SJR) of $37.9\:\mr{dB}$. In other words, we focus on an energy range that is negligible for traditional barrage jamming.

%In each transmission slot, the channel matrices $H$, $G$, and $F$ are generated randomly. The distances from the jammer to the transmitter and receiver are the same. We compare the BER of both the data packet and the ACK packet. Data packets are affected by barrage jamming and pilot jamming; and ACK packets are affected by ACK jamming and pilot jamming aiming at the preceding data frame. 

Due to space limit, we only show the results for BPSK-modulated data packets in Fig. \ref{fig:ber}. We use bit error rate (BER) as the metric, since it is easy to measure, and is directly decided by SINR, the metric used in both Theorem \ref{thm1} and \ref{thm2}. Fig. \ref{fig:ber} shows that with small pilot length, pilot jamming significantly outperforms barrage jamming. We also notice that the performance gain decreases with increasing  pilot length -- with pilot length of $512$ symbols, pilot jamming becomes similar to barrage jamming on performance.  This matches the result predicted by Theorem \ref{thm1}, which states that pilot jamming is more energy-efficient than barrage jamming when ${K < \sqrt{DM}}$ -- in this case, the theoretical crossover point is approximately $495$. The effect of pilot jamming on the following ACK packet can also be verified in Fig. \ref{fig:ber}. For $K = 4$, pilot jamming results in higher BER than ACK jamming with the same energy. However, the advantage becomes negligible for $K = 16$ and ACK jamming becomes better for $K = 128$.  This observation validates Theorem \ref{thm2}, which states that pilot jamming is better than ACK jamming if ${K < \sqrt{A \theta_G N^2 / (L \theta_F)}}$ -- the point in this case is $32$. 

Fig. \ref{fig:ber} does not precisely match Theorem \ref{thm3}, but we argue that the lower bound of PER is not always reflected well by BER. Actually, with larger jamming energy (not shown due to space limit), the results in BER matches Theorem \ref{thm3} much better, suggesting that Theorem \ref{thm3} is more accurate with high jamming energy.

% Theorem \ref{thm3} predicts that, with the simulation setting, ACK jamming always yields a higher PER lower bound for ACK packets than barrage jamming for data packets, and yields a higher PER lower bound on ACK packets than pilot jamming for data packets when $K < 59.9$. This trend is not observable from the results in Fig. \ref{fig:ber}. But, we argue that the lower bound of PER is not a linear function of BER. Actually, with larger jamming energy (not shown due to space limit), the results in BER agrees with Theorem \ref{thm3} much better, indicating that the lower bound is tighter with smaller (expected) SINR.

%\textbf{This suggests that the lower bound on the BER may not be tight among the jamming strateg, and explains why the enhanced exploration should direct to ACK jamming with a certain probability $\epsilon$ even if the condition is met.}

\vspace{-2mm}
\subsection{Experimental Testbed Evaluation}\label{sec:testbed}
We have implemented a $2\times2$ MIMO system described in Section \ref{sec:mod} using 4 USRP N210s and 1 USRP X310. For the forward (data) transmission, we use the I/Q data from the USRPs to evaluate different jamming schemes. To better control the jamming energy, we manually add barrage and pilot jamming signals to the received samples. The signal processing is run in Matlab, with BER computed as the metric. Due to the limitation on computation speed, we are unable to perform the backward (ACK) transmission in real-time, immediately after a data packet is received. We emulate the ACK transmission and ACK jamming by leveraging channel information estimated from the received samples to perform beamforming, and compute the received ACK signal with jamming signal added. The BER is computed after signal processing of the ACK packet.

\begin{figure}[!h]
    \centering
    \includegraphics[width=\columnwidth]{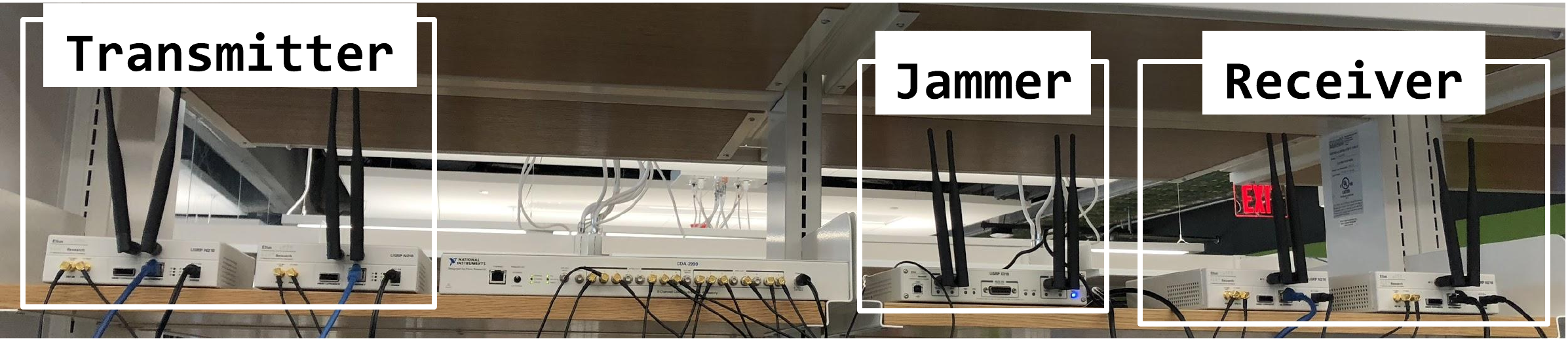}
    \caption{MIMO Experimental Testbed.}
    \label{fig:testbed}
\end{figure}

\begin{figure*}
    \centering
    \begin{tabular}{cccc}
        \includegraphics[width=0.23\textwidth]{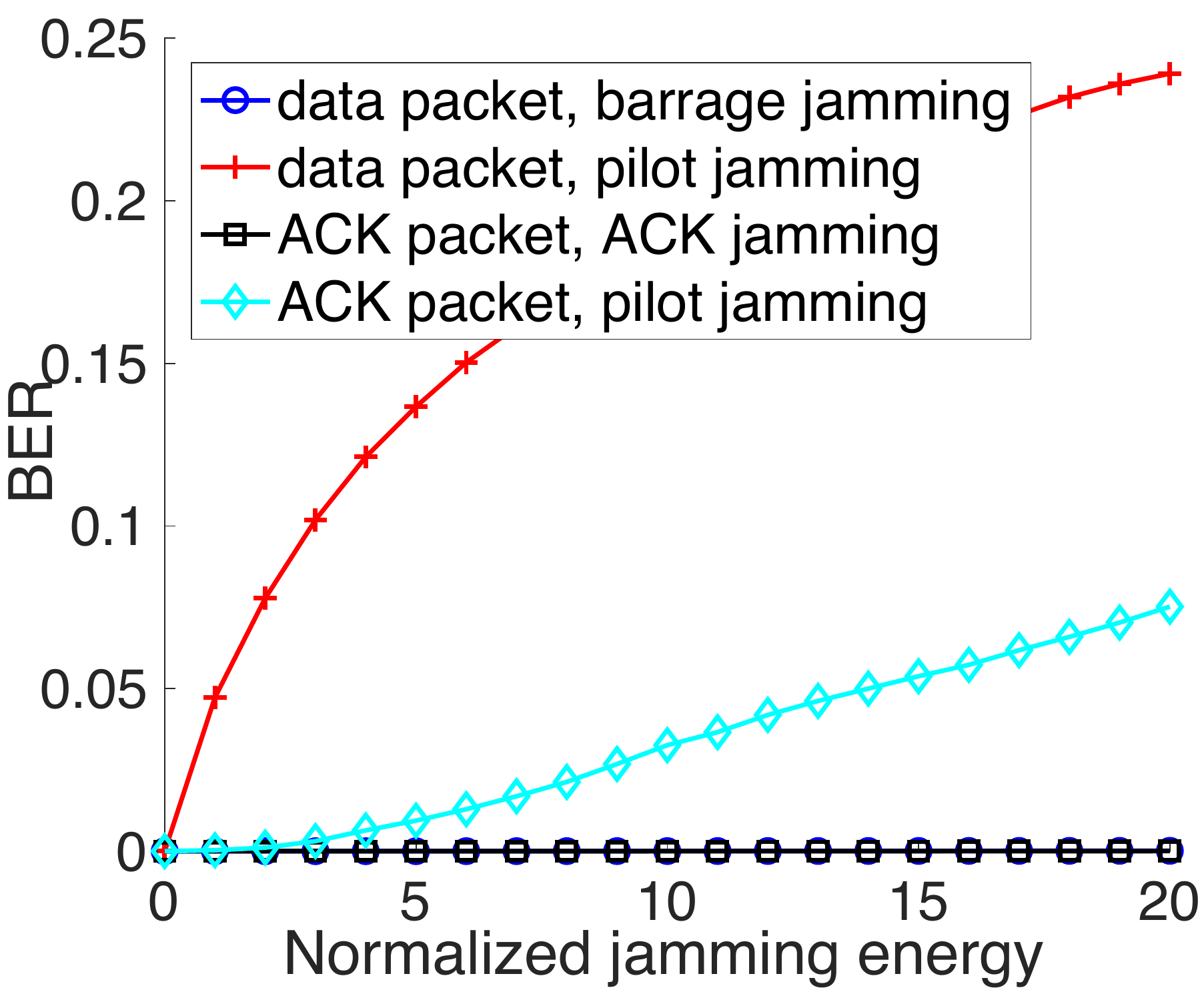} & 
        \includegraphics[width=0.23\textwidth]{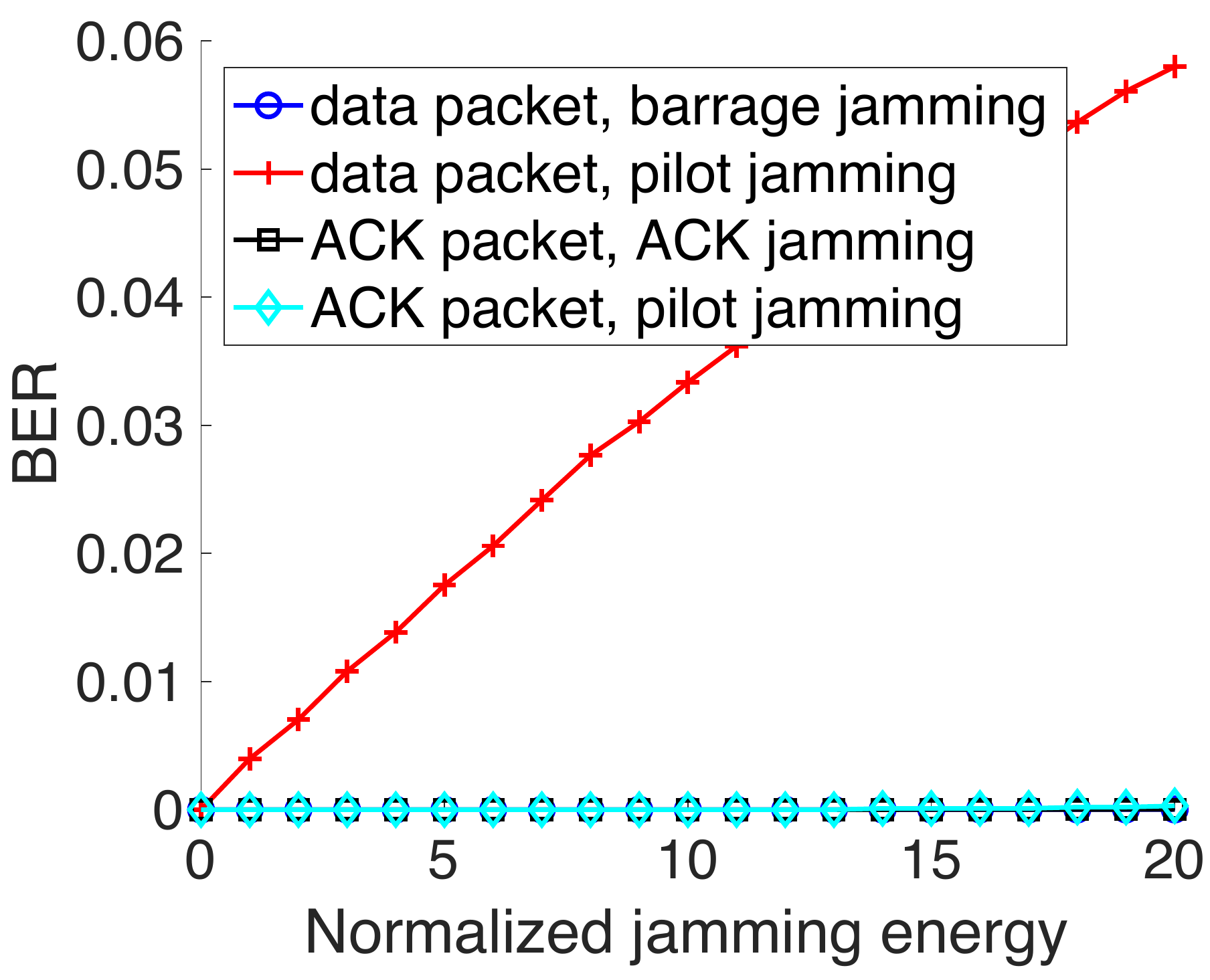} & 
        \includegraphics[width=0.23\textwidth]{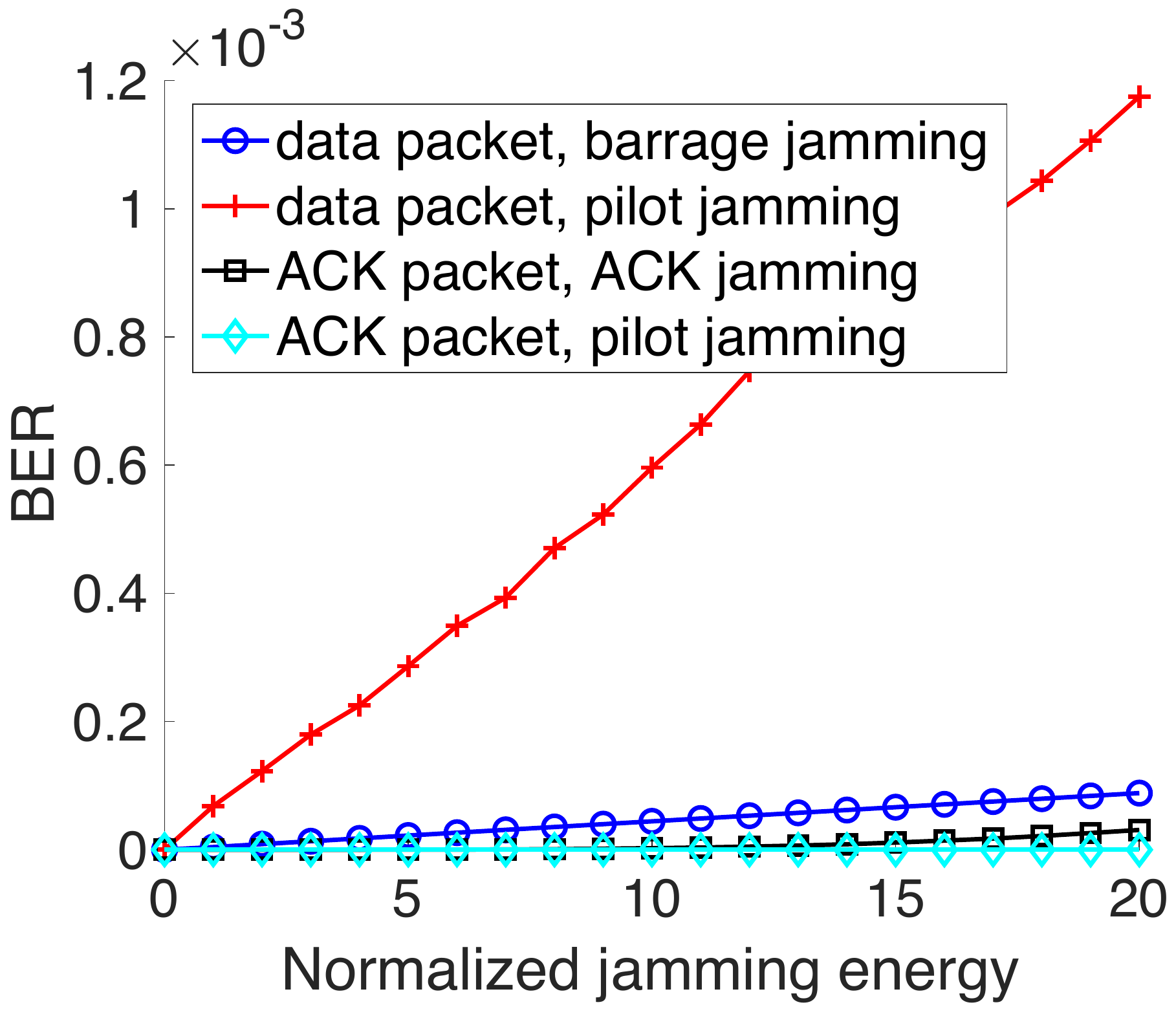} & 
        \includegraphics[width=0.23\textwidth]{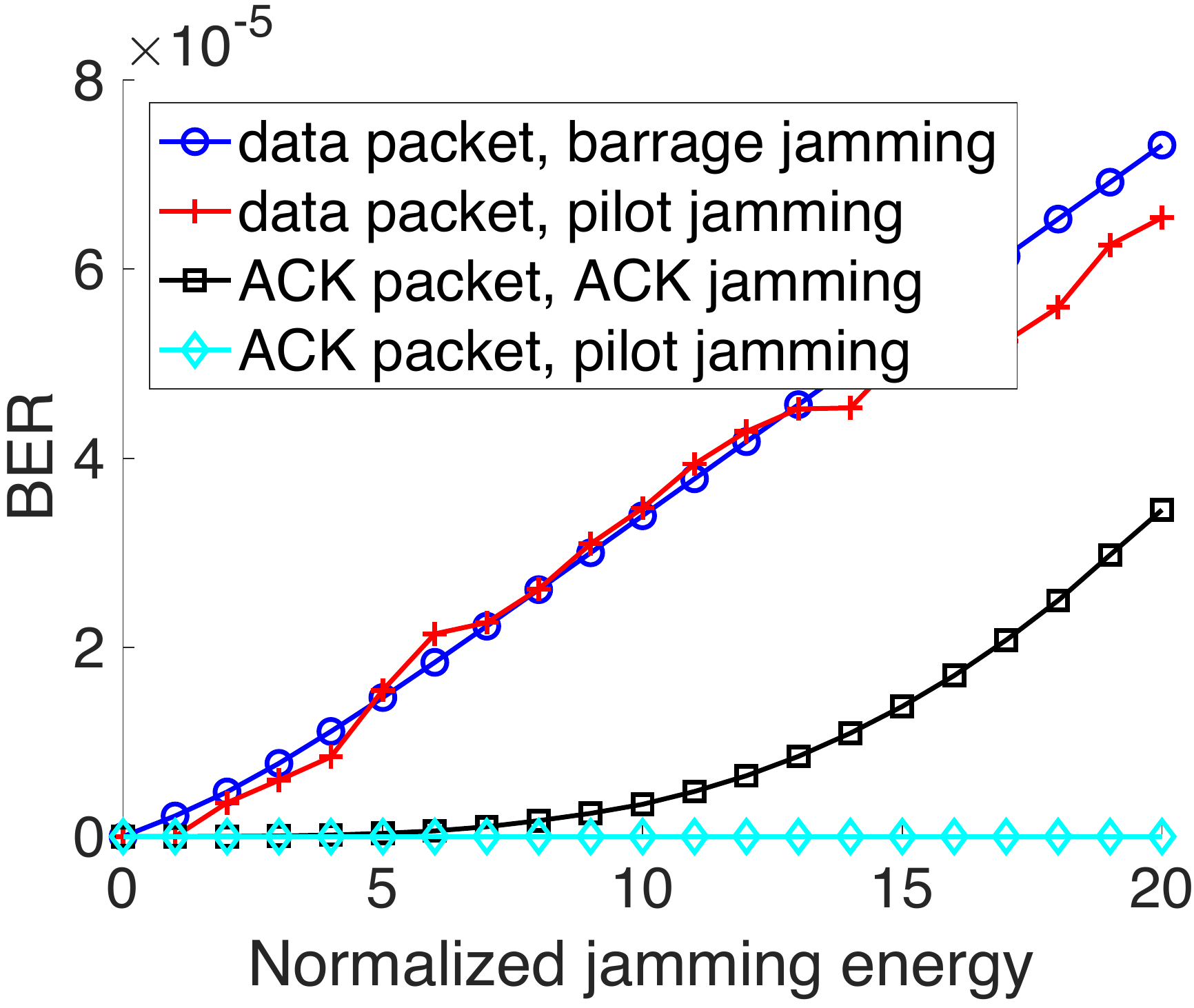} \\
        \vspace{-2mm}
        \small (a) & \small (b) & \small (c) & \small (d)
    \end{tabular}
    \vspace{-2mm}
    \caption{\small BER for different jamming schemes with pilot length $K$ = (a) $4$; (b) $16$; (c) $128$; (d) $512$. Modulation: BPSK}
    \label{fig:ber}
\end{figure*}

\begin{figure*}
    \centering
    \begin{tabular}{ccc}
    \includegraphics[width=0.25\textwidth]{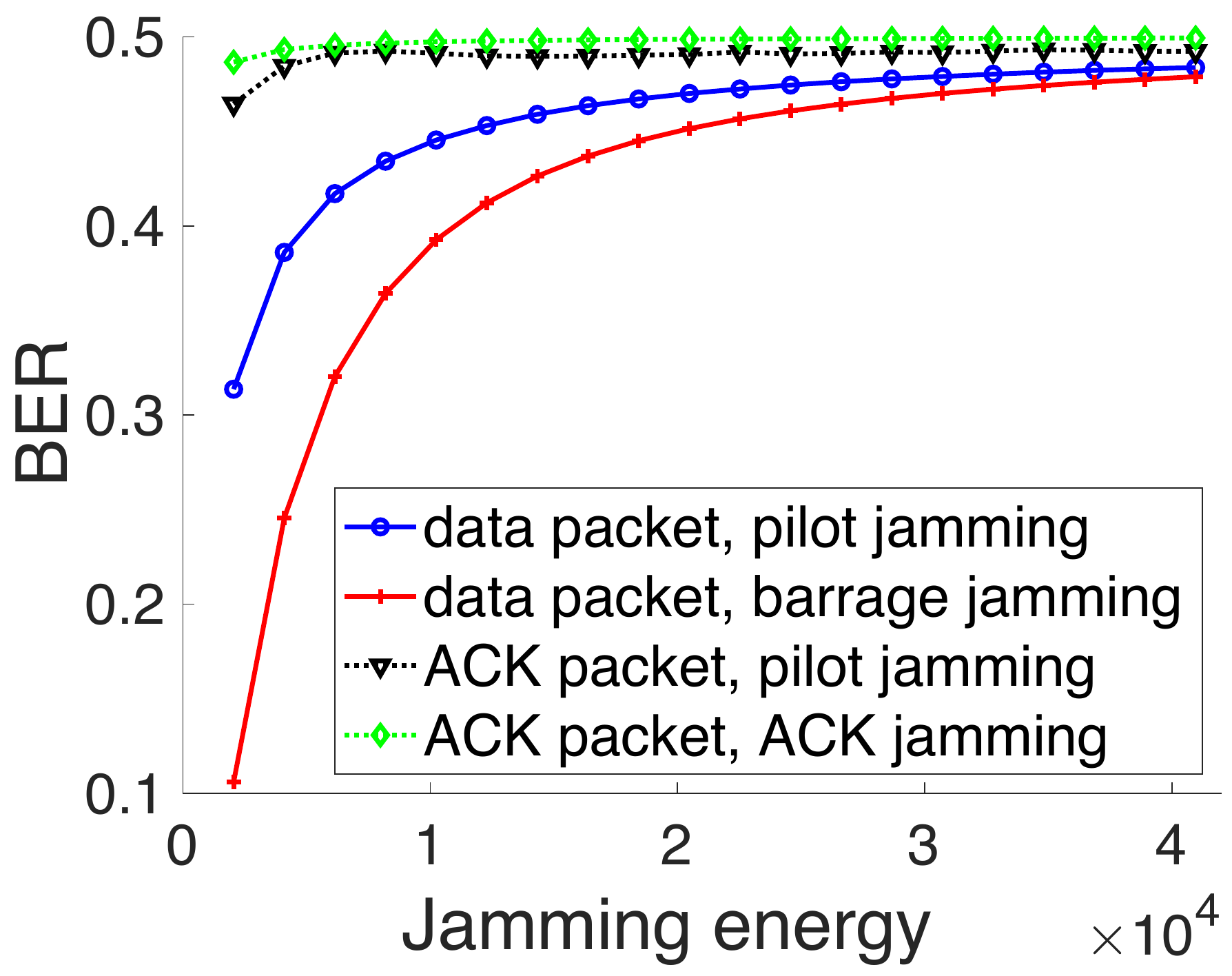} &
    \includegraphics[width=0.25\textwidth]{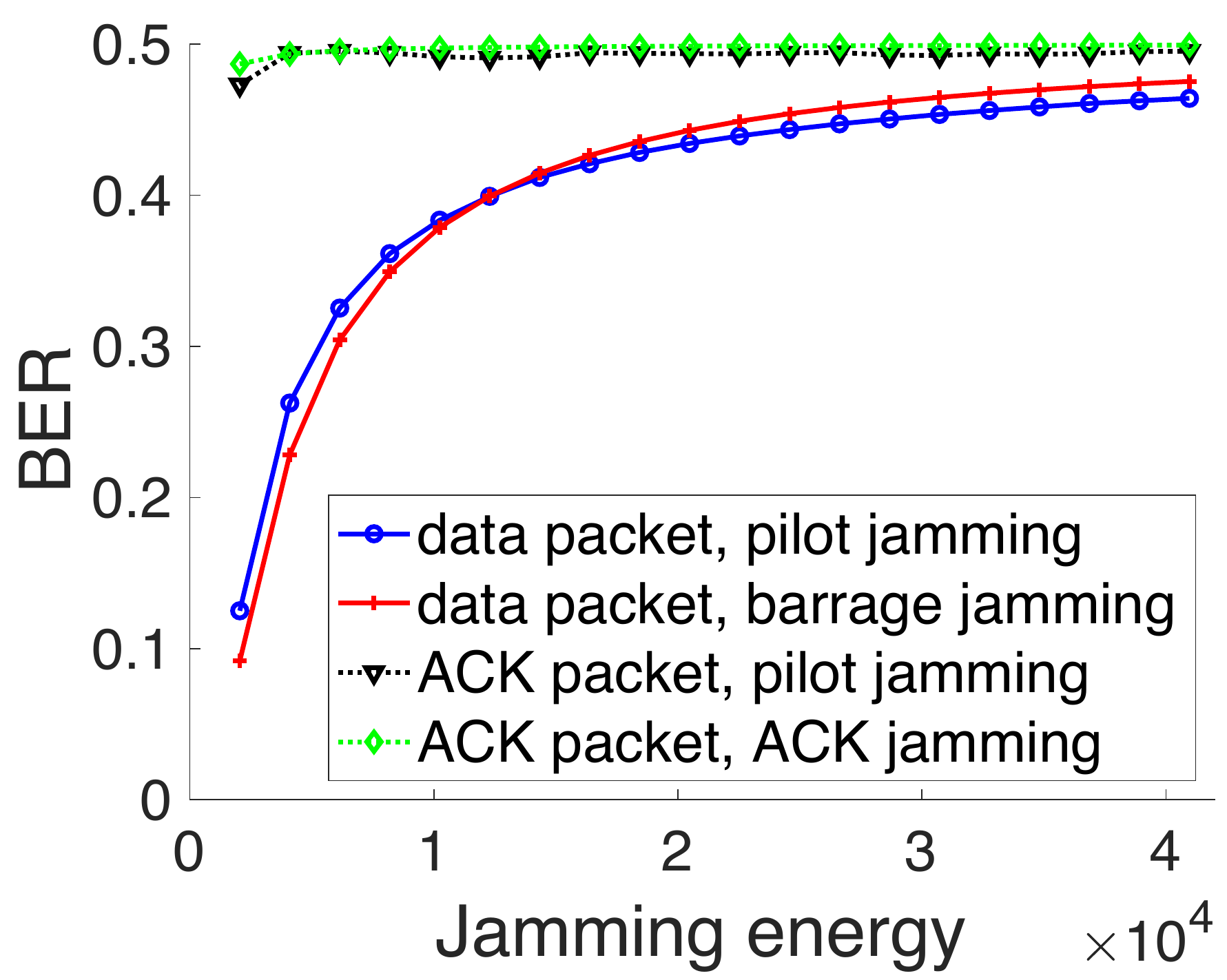} &
    \includegraphics[width=0.25\textwidth]{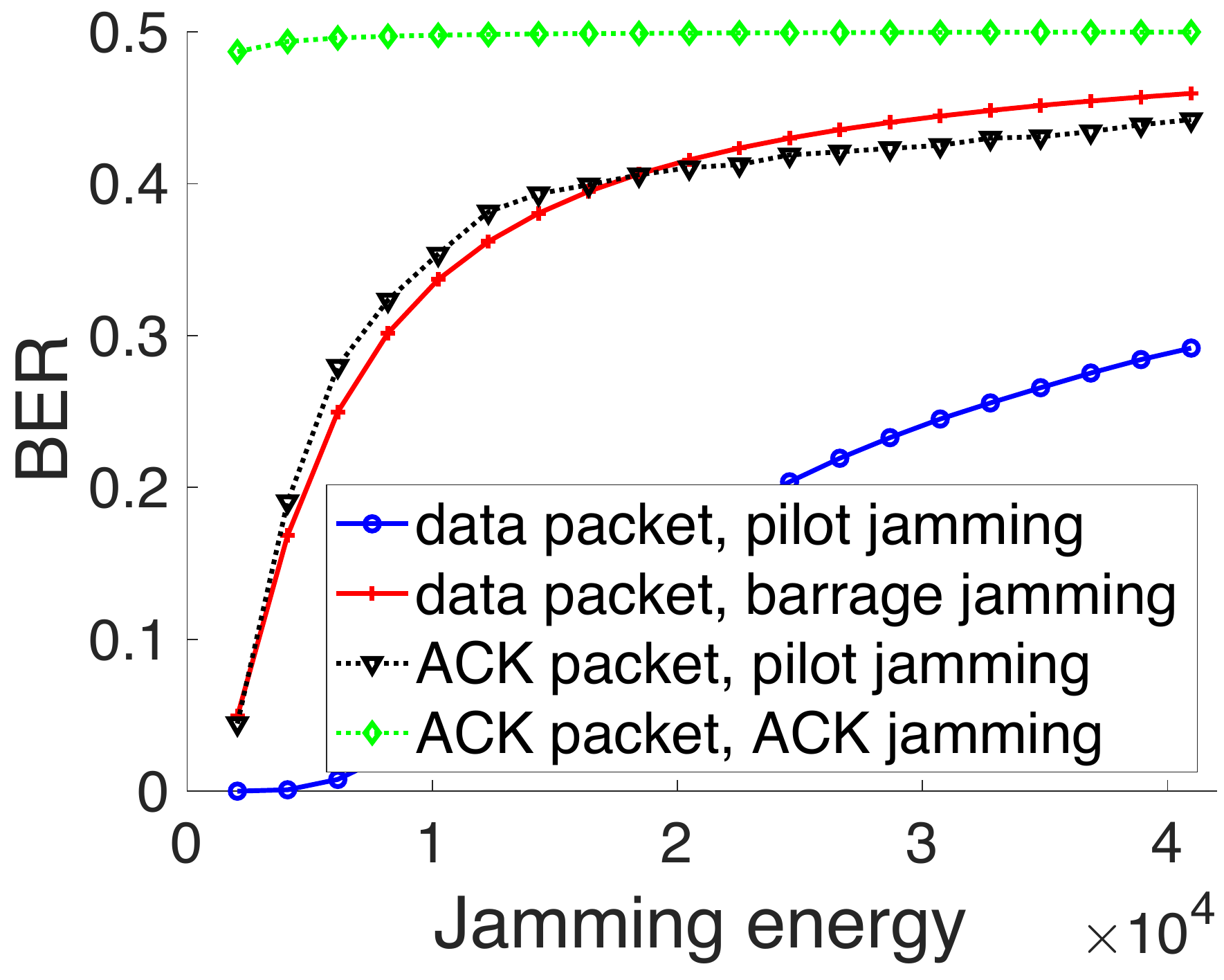}\\
    \vspace{-2mm}
    \small (a) & \small (b) & \small (c)
    \end{tabular}
    \vspace{-2mm}
    \caption{\small BER as function of jamming energy (normalized to pre-amplified symbol energy) for pilot length (a) 64; (b) 128; and (c) 512.}
    \label{fig:testbed_BER}
\end{figure*}

For the sake of simplicity, we use QPSK modulation for both data and ACK transmission. Data packets are long $1024$ symbols, while the pilot length are long $64, 128,$ and $512$ symbols -- we choose long pilots because they are also used for synchronization. The length of ACK is set to $128$ symbols. We explore $20$ levels of jamming energy, such that the interference level at the receiver side is significant enough to create bit errors. Specifically, the jamming energy is set to $[0.5:0.5:10] \times 4096$, normalized by the \emph{pre-amplified} symbol energy for the data transmission. 

Fig. \ref{fig:testbed_BER} confirm simulations results. Indeed, for data packets, pilot jamming prevails with short pilots, but becomes less efficient than barrage jamming when the pilot length is increased to a certain level. Although the results do not match precisely Theorem \ref{thm1} on the crossover point, we point out that Theorem \ref{thm1} is derived based on the assumption of negligible channel estimation errors, which does not hold in the testbed experiments. Furthermore, the relation between pilot jamming efficiency and pilot length is still valid. For ACK jamming, we observe that ACK jamming effectively destructs ACK transmission, achieving a BER around $0.5$.

% \section{Performance Evaluation} 
% In this section, we verify the validity of the theorems derived in Section \ref{sec:ana} by simulating realistic jamming effects for the model in Section \ref{sec:mod}. Furthermore, we evaluate through simulations the effectiveness of the learning-based jamming algorithm in Section \ref{sec:alg}. Finally, we implement a testbed of the described MIMO system using USRPs, and use the collected real-world data to verify the theoretical findings.

%Similar results are observed for other modulations; results are omitted because of space limitations.

\subsection{Evaluation of Learning Algorithm}\label{sec:eva}

In our experiments, we assume there are two routes available to the transmitter, while the set of MCSs includes BPSK-1/2, QPSK-9/16, 16QAM-3/4, and 64QAM-3/4 -- we will refer to them as MCS 1 to 4. All the MCSs comply to the PER-SINR function in (\ref{eq:per}), with the parameters $a_z$ and $b_z$ borrowed from \cite{Liu2004TWC}. The state set is then composed of $8$ states, indexed as in Table \ref{tbl:state_set}.

\begin{table}[!h]
    \centering
    \begin{tabular}{|c|c|}
        \hline
        \cellcolor{blue!10}\textbf{State Index} & \cellcolor{green!10}\textbf{Description} \\
        \hline
        1 - 4 & Route 1, with MCS 1 - 4 \\        \hline
        5 - 8 & Route 2, with MCS 1 - 4 \\        \hline
        \hline
       \cellcolor{blue!10} \textbf{Action Index} & \cellcolor{green!10}\textbf{Description} \\
        \hline
        1 - 20 & Barrage, energy 1 - 20 \\        \hline
        21 - 40 & Pilot, energy 1 - 20, length 4 \\        \hline
        41 - 60 & Pilot, energy 1 - 20, length 16 \\        \hline
        61 - 80 & Pilot, energy 1 - 20, length 128 \\        \hline
        81 - 100 & Pilot, energy 1 - 20, length 512 \\        \hline
        101 - 120 & ACK, energy 1 - 20 \\        \hline
        \hline
    \end{tabular}
        \caption{State and action sets.\vspace{-0.5cm}}
    \label{tbl:state_set}
\end{table}

The jamming normalized energy is set from $10$ to $200$ with a step of $10$. The jammer needs to decide the pilot jamming length $T_p$, since the jammer is not aware of the real value, which is fixed to $128$ in the simulation. We discretize the pilot jamming length to values of $4$ levels, $4$, $16$, $128$, and $512$. We simulate transmission with dynamic MCS and route adaptation (\textit{i.e.}, the transmitter adjusts its MCS and route according to the link throughput). To better illustrate the algorithm convergence in different states, we intentionally let the transmitter perform state transition every $1000$ steps. 
% \begin{table}[!h]
%     \vspace{-2mm}
%     \centering

%         \caption{Action set.}
%     \label{tbl:action_set}
% \end{table}

\begin{figure}[!h]
    \vspace{-2mm}
    \centering
    \begin{tabular}{cc}
    \includegraphics[width=0.23\textwidth]{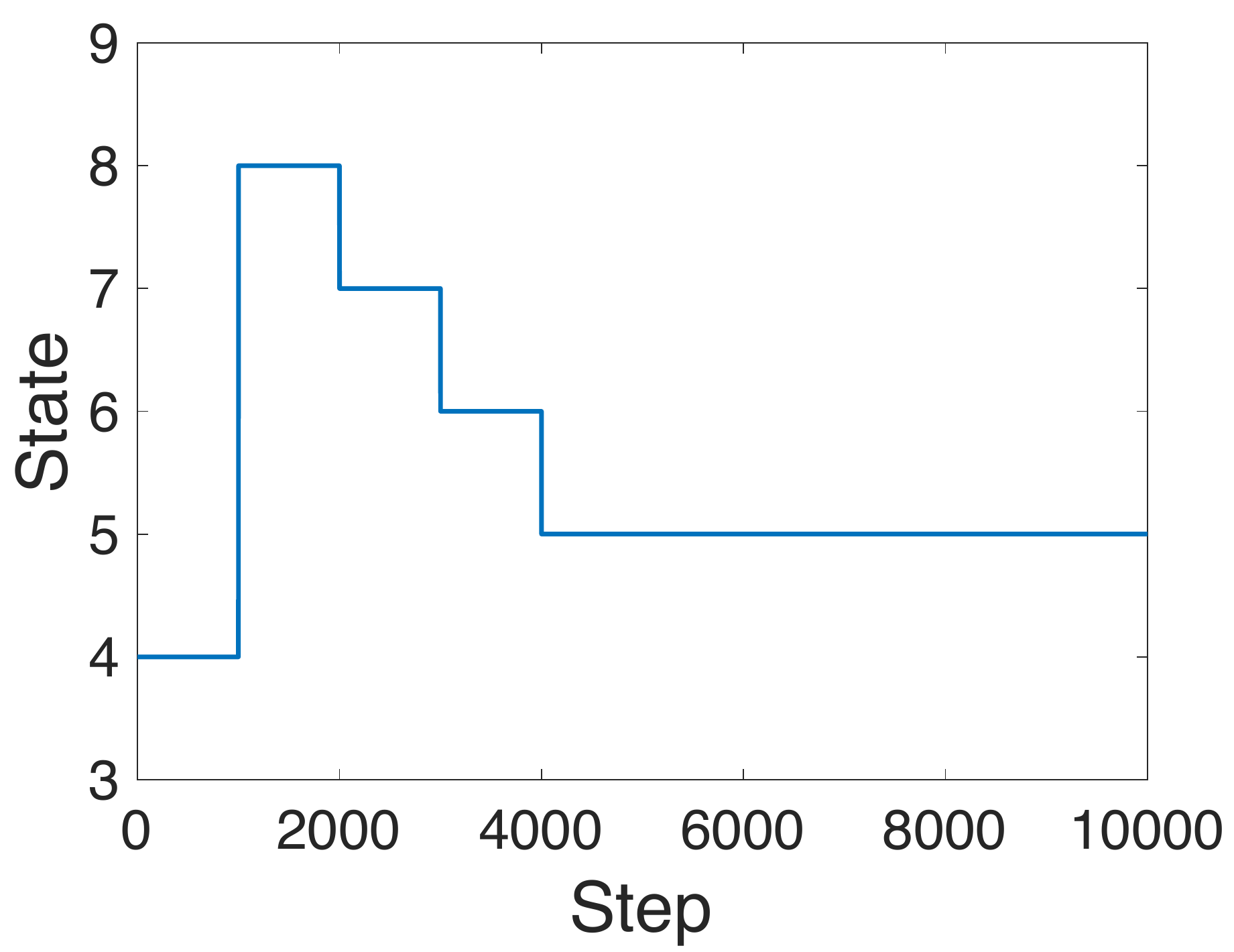} & 
    \includegraphics[width=0.23\textwidth]{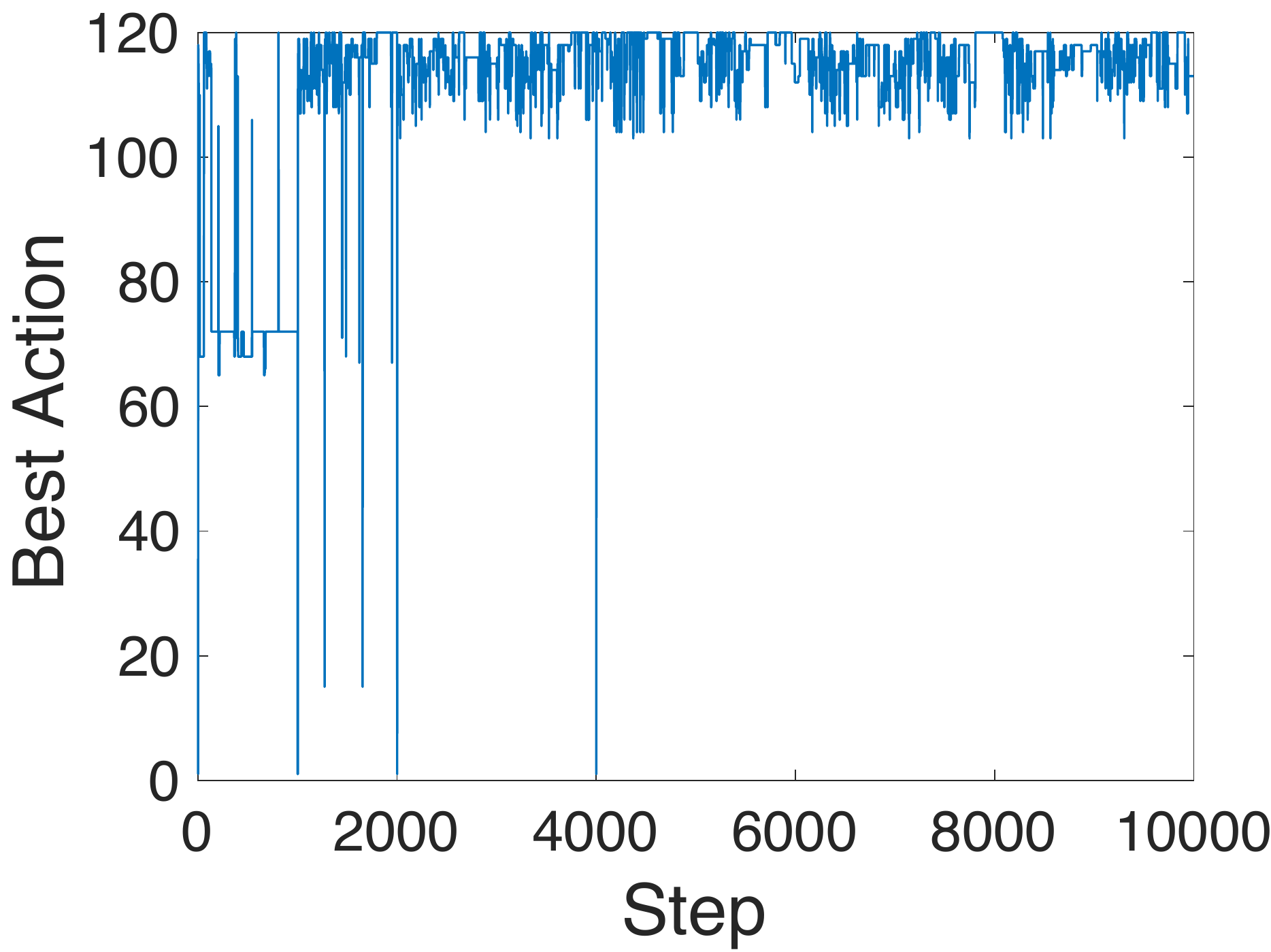} \\
    \small (a) & \small (b) \\
    \includegraphics[width=0.23\textwidth]{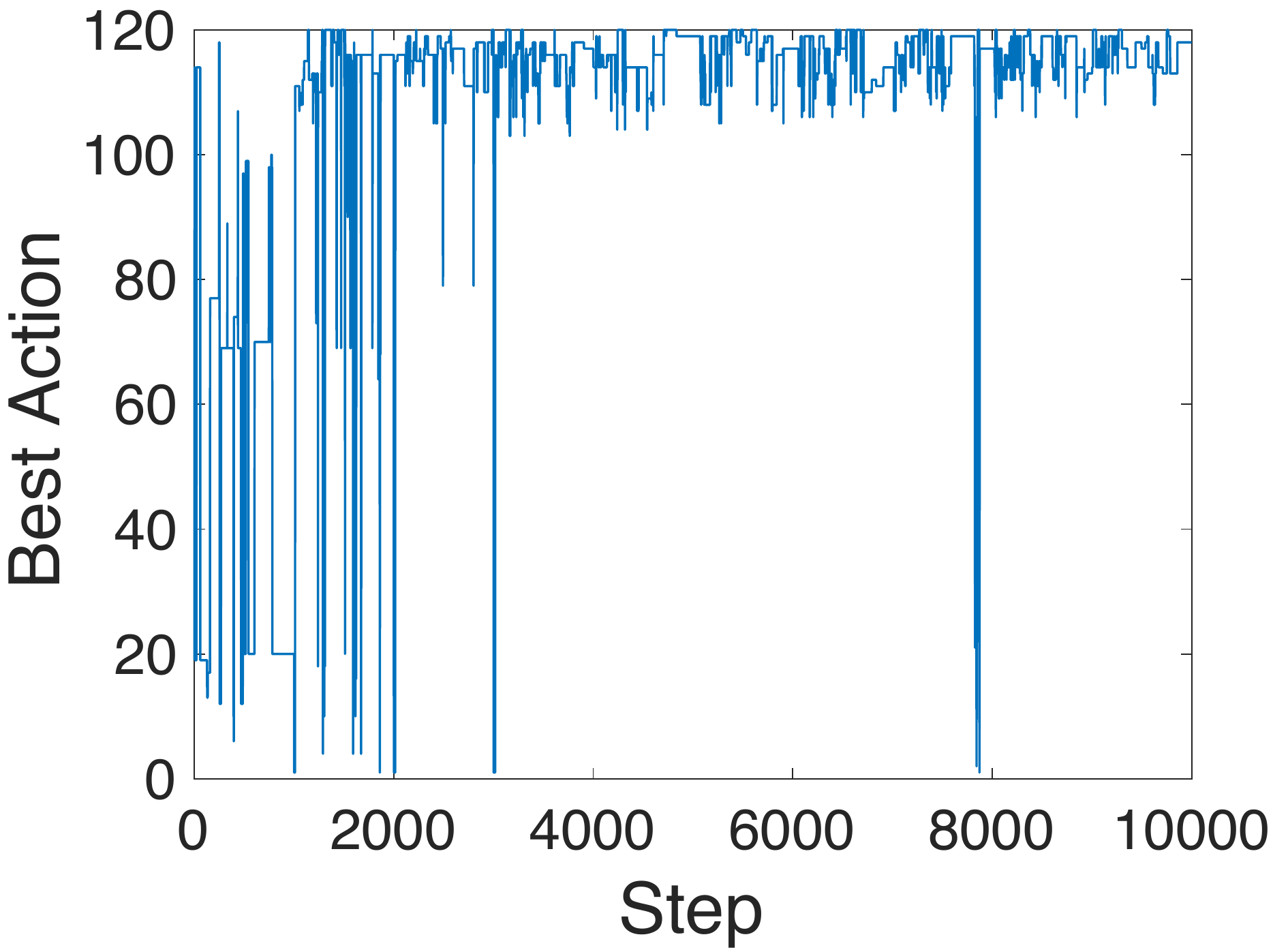} &
    \includegraphics[width=0.23\textwidth]{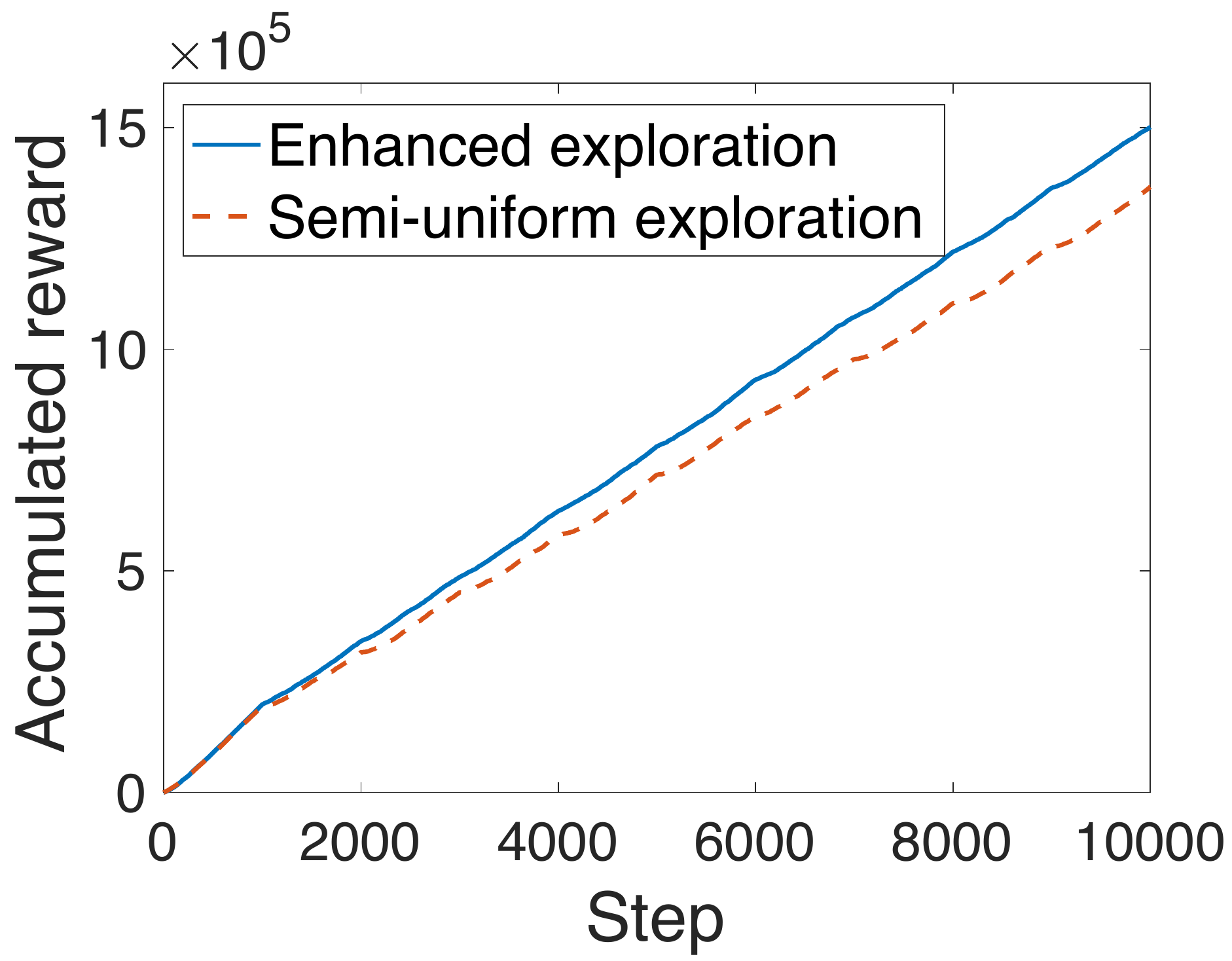} \\
    \small (c) & \small (d)\\
    \end{tabular}
    \vspace{-2mm}
    \caption{\small Learning-based jamming: (a) state transition; (b) best action, enhanced exploration; (c) best action, semi-uniform exploration; (d) accumulated reward.}
    \vspace{-2mm}
    \label{fig:q_alg}
\end{figure}

%\begin{figure}
%    \centering
%    \includegraphics[width=0.35\textwidth]{figs/q_alg_a.pdf}
%    \caption{Best action at each step, enhanced exploration.}
%    \label{fig:q_alg_a}
%\end{figure}
The resultant state transition is shown in Fig. \ref{fig:q_alg} (a), and the best action at each step is shown in Fig. \ref{fig:q_alg} (b). We can observe that the transmitter starts with route 1 and the highest MCS of 64QAM-3/4. Since receiver 1 is close to the jammer and the pilot length $K = 128$ satisfies (\ref{eq:cond1}), the best action is pilot jamming with length $128$. This is confirmed by looking at the best action in the initial steps in Fig. \ref{fig:q_alg} (b). After step $1000$, the transmitter switches to route 2, and since receiver 2 is far away while the transmitter is close to the jammer, the best action becomes ACK jamming, agreeing with the results shown in Fig. \ref{fig:q_alg} (b). The success of the jammer in degrading throughput is proven by the state adaptation process, which shows that the transmitter adapts a more reliable MCS in each update, until the most reliable one, \textit{i.e.}, BPSK-1/2 is used. The jammer's best action remains ACK jamming since ACK packets are not affected by the MCS adaption. To show the benefits of the enhanced exploration method discussed in Section \ref{sec:alg}, we run simulations with the same settings and semi-uniform exploration. State adaptation is the same as in Fig. \ref{fig:q_alg} (a), but the convergence is slower. Fig. \ref{fig:q_alg} (c) concludes that the fluctuation in best action is clearly larger than that for enhanced exploration before step $2000$. The benefit of enhanced exploration is clearer in Fig. \ref{fig:q_alg} (d), where the accumulated reward is shown for both exploration methods.

\section{Conclusions} \label{sec:con}
In this paper, we have identified a set of potential targets for advanced jamming. Then, we have formally modeled jamming schemes aimed at each of these vulnerabilities, and conducted a rigorous analysis  resulting in insightful theorems that unveil optimal jamming strategies in scenarios of interest. Them, we have designed a reinforcement learning based algorithm that allows the jammer to adapt its jamming strategy to dynamic environments. The theorems were used to enhance the efficiency of action exploration in the learning process. We have verified the theorems and proved the effectiveness of the proposed algorithm through extensive simulations and experiments.

\section*{Acknowledgements}

We sincerely thank our shepherd Xiaowen Gong and the anonymous reviewers for their constructive feedback, which has helped increase significantly the quality of our manuscript. This material is based upon work funded by the Air Force Research Laboratory (AFRL) under Contract No. FA8750-15-3-6001-NU and by the National Science Foundation (NSF) under Grant CNS-1618727. Any opinions, findings and conclusions or recommendations expressed in this material are those of the authors and do not necessarily reflect the views of the AFRL or the NSF.

\end{document}